\documentclass[twoside,12pt,reqno]{amsart}
\usepackage{a4wide}
\usepackage{latexsym}
\usepackage{amsfonts}
\usepackage{amssymb}
\usepackage{amsthm}
\usepackage{amsmath}              
\usepackage[noadjust]{cite}       
\usepackage{xypic}
\usepackage{wasysym}              
\newcommand{\cl}{C \kern -0.1em \ell}     

\newcommand{\Pin}{\mathbf{Pin}}

\newcommand{\Mat}{{\rm Mat}}

  \newcommand{\be}{{\bf e}}

  \newcommand{\bv}{{\bf v}}

\newcommand{\BK}{\mathbb{K}}

\newcommand{\BC}{\mathbb{C}}
\newcommand{\BR}{\mathbb{R}}
\newcommand{\BH}{\mathbb{H}}
\newcommand{\BZ}{\mathbb{Z}}

\newcommand{\spn}{\mbox{\rm span}}

\def\dim{\hbox{\rm dim\,}}

\newcommand{\beq}{\begin{equation}}
\newcommand{\eeq}{\end{equation}}

\def\dim{\hbox{\rm dim\,}}

\def\CLIFFORD{\mbox{\bf \tt CLIFFORD}}

\input{diagxy}   

\newcommand{\ed}{\end{document}}
\def\ve{\varepsilon}

\newcommand{\ta}[2]{#1_#2\tilde{\phantom{.}}}
\newcommand{\cb}[1]{\mathcal{#1}}

\newcommand{\iu}{{\underline{i}}}
\newcommand{\ju}{{\underline{j}}}
\newcommand{\ku}{{\underline{k}}}

\newcommand{\End}{\mathrm{End}}

\newcommand{\cml}{{\mathrm{cmul}}}

\newcommand{\tp}{\ta{T}{\ve}}

\newcommand{\clkj}[3]{c_{#1,#2}^{#3}}
\newcommand{\phm}{\phantom{-}}

\theoremstyle{plain}
\newtheorem{theorem}{Theorem}
\newtheorem{corollary}{Corollary}
\newtheorem{lemma}{Lemma}
\newtheorem{proposition}{Proposition}
\theoremstyle{definition}
\newtheorem{definition}{Definition}
\newtheorem{example}{Example}

\begin{document}
\title[On the Transposition Anti-Involution in Real Clifford Algebras II: ...]
{On the Transposition Anti-Involution in\\Real Clifford Algebras II:\\
Stabilizer Groups of Primitive Idempotents}
\author[Rafa\l \ Ab\l amowicz]{Rafa\l \ Ab\l amowicz}

\address{%
Department of Mathematics, Box 5054\\
Tennessee Technological University\\
Cookeville, TN 38505, USA}
\email{rablamowicz@tntech.edu}

\author[Bertfried Fauser]{Bertfried Fauser}
\address{%
School of Computer Science\\
The University of Birmingham\\
Edgbaston-Birmingham, W. Midlands, B15 2TT\\
United Kingdom}
\email{B.Fauser@cs.bham.ac.uk}

\subjclass{11E88, 15A66, 16S34, 20C05, 68W30}

\keywords{correlation, dual space, exterior algebra, grade involution,  group
action, group ring, indecomposable module, involution, left regular
representation, minimal left ideal, monomial order, primitive idempotent,
reversion, simple algebra, semisimple algebra, spinor, stabilizer, transversal,
universal Clifford algebra}

\date{May 13, 2010} 


\begin{abstract}
In the first article of this work~\cite{part1} we showed that real Clifford algebras $\cl(V,Q)$  posses a unique transposition anti-involution $\tp$. There it was shown that the map reduces to reversion (resp. conjugation) for any Euclidean (resp.  anti-Euclidean) signature. When applied to a general element of the algebra,  it results in transposition of the associated matrix of that element in the  left regular representation of the algebra. In this paper we show that, depending on the value of $(p-q) \bmod 8$, where $\ve=(p,q)$ is the signature of $Q$, the anti-involution gives rise to transposition, Hermitian complex, and Hermitian quaternionic conjugation of representation matrices in spinor representations. We realize spinors in minimal left ideals $S=\cl_{p,q}f$  generated by a primitive idempotent~$f$. The map $\tp$ allows us to define a dual spinor space $S^\ast$, and a \emph{new spinor norm} on $S$, which is different, in general, from two spinor norms known to exist. We study a transitive action of generalized Salingaros' multiplicative vee groups $G_{p,q}$ on complete sets of mutually annihilating primitive idempotents. Using the normal stabilizer subgroup $G_{p,q}(f)$ we construct left transversals, spinor bases, and maps between spinor spaces for different orthogonal idempotents $f_i$ summing up to $1$. We classify the stabilizer groups according to the signature
in simple and semisimple cases.
\end{abstract}
\maketitle
\section{Introduction}\label{introduction}
\medskip

This paper is a continuation of~\cite{part1}. Hence, in particular, all notation is the same. Recall that in a universal real Clifford algebra $\cl_n\cong\cl(V,Q)$ of a non-degenerate quadratic real vector space of dimension $n$, we have introduced a certain transposition anti-involution $\tp$ of the algebra as a unique extension of a suitable orthogonal map $t_\varepsilon:V \rightarrow V^\flat$. Using the identification $V^\flat \cong V^\ast,$ the orthogonal map $t_\varepsilon$ can be viewed as a symmetric non-degenerate correlation on~$V$.  For the study of correlations and involutions we refer to \cite{porteous}.

Following \cite{part1}, let $\cl_n$ and $\cl_n^\ast$ denote, respectively, the universal Clifford algebra over $V$ and the universal dual Clifford 
algebra\footnote{As it is explained in \cite{part1}, this is not the linear dual Clifford algebra, which is not employed in this paper, but the Clifford algebra over the space $V^\flat\cong V^*$ for the same quadratic form $Q$.} over $V^\ast$ where $\dim_\BR V = \dim_\BR V^\ast = n.$ Let $\cb{B}$ and $\cb{B}^\ast$ be, respectively, bases in $\cl_n$ and $\cl_n^\ast$ sorted in the monomial order 
$\mathtt{InvLex}$ and let $L_u:\cl_n \rightarrow \cl_n,$ $ u \in \cl_n$, be the left multiplication operator. In \cite[Prop. 2(v-vi)]{part1} it was  shown that if $[L_u]$ (resp. $[L_{\tp(u)}]$) is the matrix of the operator $L_u$ (resp. $L_{\tp(u)}$) relative to the basis $\cb{B}$, then these matrices are related via the matrix transposition. As a consequence, the anti-involution $\tp$ applied to $u$ results in the transposition of the matrix $[L_u]$ in the left-regular representation  $L_u: \cl_n\rightarrow \cl_n$.

In Examples~4 and 5 in~\cite{part1} it was also shown that the matrix transposition induced by the anti-involution $\tp$ resulted also in matrix transposition in real spinor representations of simple Clifford algebras $\cl_n \cong \Mat(\BR,N)$, that is, in signatures $(p,q)$ such $(p-q) \neq 1 \bmod 4$ and $(p-q) = 0,1,2 \bmod 8.$ The aim of this paper is a detailed study of $\tp$ for spinor representations in all signatures.    

In Section~\ref{action1} we study spinor representation of central \emph{simple Clifford algebras} $\cl_{p,q}$ in a minimal left ideal $S=\cl_{p,q}f$ generated by a primitive idempotent $f$ and consider 
generalized Salingaros' multiplicative vee groups $G_{p,q}$ of order $2^{1+p+q}$. It is well-known that $G_{p,q}$ acts via conjugation on $\cl_{p,q}$. We find a stabilizer group $G_{p,q}(f)$ of $f$ and show that:
\begin{itemize}
\item[\textbf{(i)}] $G_{p,q}(f)$ is normal in $G_{p,q}$ and of order $2^{1+p+r_{q-p}}$.
\item[\textbf{(ii)}] The set $\cb{F}$ of all $N=2^k$ primitive mutually annihilating idempotents $f_i$, determined by $f$, constitutes one orbit of $f$ under the (transitive) action of $G_{p,q}$.
\item[\textbf{(iii)}] Monomials $m_i$ in a (non-canonical) left transversal of $G_{p,q}(f)$ together with $f$ determine a spinor basis in the left ideal $\cl_{p,q}f$.
\item[\textbf{(iv)}] The division rings  $\mathbb{K}=f\cl_{p,q}f$ are $G_{p,q}$-invariant.
\item[\textbf{(v)}] $G_{p,q}$ permutes the spinor basis elements modulo the commutator subgroup $G_{p,q}'$ by acting on them via the left multiplication.
\end{itemize}
We classify all stabilizer groups and analyze their structure in dimension up to nine. We recognize that the anti-involution  $\tp$ is similar to a map known in the theory of group rings as 
$\ast: k[G]\rightarrow k[G]$ sending  $\sum_{x \in G} a_x x \mapsto \sum a_x x^{-1}$~\cite{passman}.\footnote{It may interest the curious reader that this map is in fact the antipode on the Hopf algebra over the group algebra $k[G]$.} This map allows one to define a spinor norm on $\cl_{p,q}f$ which is different, in general, from $\beta_{\pm}$  norms on spinor spaces known to exist~\cite{lounesto}. This spinor norm is invariant under a group $G_{p,q}^\varepsilon$ defined in this section. Finally, we show that the real anti-involution  $\tp$ of $\cl_{p,q}$ is responsible for transposition, Hermitian complex, and Hermitian quaternionic conjugation of a matrix $[u]$ in spinor representation, where $u \in \cl_{p,q}$, depending on the value of $(p-q) \bmod 8.$ 

In Section~\ref{action2}, we extend results from simple to \emph{semisimple Clifford algebras} and realize their faithful spinor representation in $S \oplus \hat{S}$. We show that:
\begin{itemize}
\item[\textbf{(i)}] $G_{p,q}(f) = G_{p,q}(\hat{f})$ for any primitive idempotent $f$.
\item[\textbf{(ii)}] $G_{p,q}(f)$ is normal in $G_{p,q}$ and its order is $2^{2+p+r_{q-p}}$. 
\item[\textbf{(iii)}] Since the center of a semisimple Clifford algebra is non-trivial and includes a unit pseudoscalar $\be_{1 \cdots n}$ which squares to $1$, there are two orthogonal central idempotents
$\frac12(1 \pm \be_{1 \cdots n})$ adding to $1$. Then, the set $\cb{F}$ of $2^{k}$ mutually annihilating primitive idempotents adding to the unity partitions into two sets $\cb{F}_1$ and $\cb{F}_2$, each with $N=2^{k-1}$ elements. As expected, the idempotents in each set add up to one of the two central idempotents. 
\item[\textbf{(iv)}] Like in the simple case, $k = q - r_{q-p}.$ We show that the set $\cb{F}_1$ coincides with the orbit $\cb{O}(f)$ while the set $\cb{F}_2$  coincides with the orbit $\cb{O}(\hat{f})$, each under the conjugate action of the vee group $G_{p,q}$.
\item[\textbf{(v)}] No element of $G_{p,q}$ relates the two orbits.
\item[\textbf{(vi)}] The transversal of $G_{p,q}(f)$ in $G_{p,q}$ generates, like in the simple case, a spinor basis in $S$, hence in $\hat{S}$ too, and it permutes the spinor basis elements.
\item[\textbf{(vii)}] Finally, the transversal elements act transitively via conjugation on each orbit.
\end{itemize}
The remaining results from the simple case apply to the semisimple case. In particular, the real map $\tp$ is responsible for transposition or the Hermitian quaternionic conjugation of a matrix $[u]$ for any $u \in \cl_{p,q}$. This time, we need to work with double spinor representation in $S \oplus \hat{S}$ over the double ring $\BK \oplus \hat{\BK}$, depending on the value of $(p-q) \bmod 8.$   

In Section~\ref{conclusions}, we summarize our results. In Appendix \ref{AppendC} we collect information about the stabilizer groups $G_{p,q}(f)$ for simple Clifford algebras in Tables~1, 2, and 3,  and for semisimple algebras  in Tables~4 and~5.   

\section{Action on spinor spaces in simple Clifford algebras}\label{action1}
\medskip

In this section we will show that the anti-involution $\ta{T}{\ve}$ corresponds to Hermitian conjugation of matrices in spinor representation of $\cl_{p,q}.$ In particular, it acts as Hermitian conjugation on the spinor space $S=\cl_{p,q}f$ where $f$ is a primitive idempotent. For the theory of spinor representations of Clifford algebras, we refer to~\cite{hahn,helmmicali,lounesto} while 
in~\cite{ablamowicz1998} one can find many computational examples. We also investigate the conjugate action of the Salingaros vee group $G_{p,q}$ on primitive idempotents in a simple algebra $\cl_{p,q}$ and we classify their stabilizer groups. For the necessary terminology and results of the theory of groups, we refer to~\cite{rotman}.  

We begin with summarizing background information on Clifford algebras. We denote the reals as $\mathbb{R}$, complex numbers as $\mathbb{C}$ and quaternions as $\mathbb{H}$.
\begin{theorem}
Let $\cl_{p,q}$ be the universal Clifford algebra as defined above.
\begin{itemize}
\item[(a)] When $p-q \neq 1 \bmod 4$ then $\cl_{p,q}$ is a simple algebra of dimension $2^n,$ $n=p+q$, isomorphic with a full matrix algebra $\Mat(2^k, \BK)$ over a division ring $\BK$ where 
$k = q - r_{q-p}$ and $r_i$ is the Radon-Hurwitz number.\footnote{The Radon-Hurwitz number is defined by recursion as $r_{i+8}=r_i+4$ and these initial values: $r_0=0,r_1=1,r_2=r_3=2,r_4=r_5=r_6=r_7=3.$ 
See~\cite{hahn,lounesto}.} Here $\BK$ is one of $\BR, \BC$ or $\BH$. 
\item[(b)] When $p-q = 1 \bmod 4$ then $\cl_{p,q}$ is a semisimple algebra of dimension $2^n,$ $n=p+q$ isomorphic to $\Mat(2^{k-1}, \BK) \oplus \Mat(2^{k-1}, \BK),$ $k = q - r_{q-p}$, where $\BK$ is isomorphic to $\BR$ or $\BH$ depending whether $p-q=1 \bmod 8$ or $p-q=5 \bmod 8.$ Each of the two simple direct components of $\cl_{p,q}$ is projected out by one of the two central idempotents 
$\frac12(1\pm \be_{12 \ldots n}).$
\item[(c)] Any polynomial $f$ in $\cl_{p,q}$ expressible as a product
\begin{equation}
f = \frac12(1\pm \be_{\iu_1})\frac12(1\pm \be_{\iu_2})\cdots\frac12(1\pm \be_{\iu_k})
\label{eq:f}
\end{equation}
where $\be_{\iu_i},$ $i=1,\ldots,k,$ are commuting basis monomials in $\cb{B}$ with square $1$ and $k = q - r_{q-p}$ generating a group of order $2^k$, is a primitive idempotent in $\cl_{p,q}.$ Furthermore, $\cl_{p,q}$ has a complete set of $2^k$ such primitive mutually annihilating idempotents which add up to the unit $1$ of $\cl_{p,q}.$\footnote{Two idempotents $f_1$ and $f_2$ are mutually annihilating when $f_1f_2=f_2f_1=0.$ Such idempotents are also called \textit{orthogonal}. A decomposition of the unity $1$ into a sum of mutually annihilating primitive idempotents is called for short a \textit{primitive idempotent decomposition}.}
\item[(d)] When $(p-q) \bmod 8$ is $0,1,2,$ or $3,7$, or $4,5,6$, then the division ring $\BK = f \cl_{p,q}f$ is isomorphic to $\BR$ or $\BC$ or $\BH$, and the map $S \times \BK \rightarrow S,$ $(\psi,\lambda) \mapsto \psi\lambda$ defines a right  $\BK$-linear structure on the minimal left ideal $S=\cl_{p,q}f.$ 
\item[(e)] When $\cl_{p,q}$ is simple, then the map
\begin{equation}
\cl_{p,q} \stackrel{\gamma}{\longrightarrow} \End_\BK(S), \quad u \mapsto \gamma(u), \quad \gamma(u)\psi = u \psi
\label{eq:simple}
\end{equation}
gives an irreducible and faithful representation of $\cl_{p,q}$ in $S.$
\item[(f)] When $\cl_{p,q}$ is semisimple, then the map
\begin{equation}
\cl_{p,q} \stackrel{\gamma}{\longrightarrow} \End_{\BK \oplus \hat{\BK}}(S \oplus \hat{S}), \quad u \mapsto \gamma(u), \quad \gamma(u)\psi = u \psi
\label{eq:semisimple}
\end{equation}
gives a faithful but reducible representation of $\cl_{p,q}$ in the double spinor space $S \oplus \hat{S}$ where $S =\{ u f\, |\, u \in \cl_{p,q}\}$, $\hat{S} =\{ u \hat{f}\, |\, u \in \cl_{p,q}\}$ and 
$\hat{\phantom{m}}$ stands for the grade-involution in $\cl_{p,q}.$ In this case, the ideal $S \oplus \hat{S}$ is a right $\BK \oplus \hat{\BK}$-linear structure,  
$\hat{\BK} = \{\hat{\lambda} \,| \, \lambda \in \BK \}$, and $\BK \oplus \hat{\BK}$ is isomorphic to $\BR \oplus \BR$ when $p-q=1 \bmod 8$ or to $\BH \oplus \hat{\BH} $ when 
$p-q=5 \bmod 8.$\footnote{See \cite[Chapt. 9]{porteous} for a discussion of quaternionic linear spaces, and \cite[Chapt. 12]{porteous} for double fields and their anti-involutions.}
\end{itemize}
\label{th:structure}
\end{theorem}

We begin by first observing the following basic properties of $\ta{T}{\ve}$.

\begin{lemma} Let $\be_\iu$ be any basis monomial in $\cb{B}$ where $\cb{B}$ is a basis of $\cl_{p,q}$ consisting of Grassmann basis monomials. Then,
\begin{itemize}
\item[(i)] $\ta{T}{\ve}(\be_\iu) = \be_\iu^{-1}.$
\item[(ii)] If $\be_\iu^2=1$ (resp. $\be_\iu^2=-1$) then $\ta{T}{\ve}(\be_\iu) = \be_\iu$ (resp. $\ta{T}{\ve}(\be_\iu) = -\be_\iu$).
\item[(iii)] Let $\be_{\iu_1},\be_{\iu_2},\ldots,\be_{\iu_k}$ be a set of mutually commuting basis monomials in $\cb{B}$ that square to $1$ where $k = q - r_{q-p}.$ Then, the primitive idempotent~(\ref{eq:f}) is invariant under $\ta{T}{\ve}$, that is, $\ta{T}{\ve}(f)=f.$
\item[(iv)] In Euclidean $(p,0)$ and anti-Euclidean $(0,q)$ signatures, let $\bv$ be a $1$-vector belonging to $(V,Q)$ in $\cl_{p,0}$ or $\cl_{0,q},$ respectively.  That is, let
$\bv = \sum_i^n v_i \be_i$ where $v_i \in \BR$ and $n = \dim_\BR V.$ Then, 
\begin{itemize}
\item[(a)] $\bv \tp(\bv) = \tp(\bv)\bv = \sum_i^n v_i^2$.
\item[(b)] If $\bv \tp(\bv) \neq 0,$ then $\bv^{-1} = \dfrac{\tp(\bv)}{\bv \tp(\bv)}.$ 
\end{itemize}  
\end{itemize}
\label{properties1}
\end{lemma}

Let $f$ be a primitive idempotent of the form (\ref{eq:f}) and let $S = \cl_{p,q}f$ be a minimal left ideal. We make the following observations:

1. Let $\lambda \in \BK = f\cl_{p,q}f.$ Then, it is easy to see that $\lambda f = f \lambda$ since $f^2=f.$  However, it should be observed that in general, $(uf)\lambda \neq \lambda (uf)$ where 
$uf \in S$ and $u \in \cl_{p,q}.$

2. The division ring $\BK = f \cl_{p,q}f$ is a real subalgebra of $\cl_{p,q}$ of dimension $1$, $2$, or $4$ depending on the value of $(p-q) \bmod 8.$ In particular,
\begin{equation}
\BK = f \cl_{p,q}f = 
\begin{cases} 
    \spn_\BR \{ f \}                                   & \textit{if $p-q =0,1,2 \bmod 8;$} \\
    \spn_\BR \{ f, \be_\iu f \}                        & \textit{if $p-q =3,7 \bmod 8;$} \\
    \spn_\BR \{ f, \be_\iu f, \be_\ju f, \be_\ku f  \} & \textit{if $p-q =4,5,6 \bmod 8;$}
\end{cases}
\label{eq:threecases}
\end{equation}
where $\be_\iu,\be_\ju,\be_\ku$ are basis monomials in $\cb{B}$ that satisfy the required relations:
\begin{equation}
\be_\iu^2 = \be_\ju^2 = \be_\ku^2 = -1 \quad \mbox{and} \quad 
\be_\iu \be_\ju  = -\be_\ju \be_\iu  = \be_\ku\quad \mbox{(cyclically)}. 
\label{eq:rels}
\end{equation}
Furthermore, due to the first observation, the basis elements $\{ 1, \be_\iu, \be_\ju, \be_\ku  \}$ which span $\BK$ over $\BR$ modulo $f$, commute with $f$. 

3. Thus, $\ta{T}{\ve}$ acts on $\BK = f\cl_{p,q}f$ as an anti-involution. In the following, $\lambda_i \in \BR,$ $b_i$ are the basis elements in $\{ 1, \be_\iu, \be_\ju, \be_\ku \}$, and, to simplify notation, we set $\overline{\lambda} = \ta{T}{\ve}(\lambda)$ whenever $\lambda \in \BK$.
\begin{equation}
\ta{T}{\ve}: \BK \rightarrow \BK, \quad 
\lambda=fuf= \sum \lambda_i b_i\mapsto 
\overline{\lambda} = f \ta{T}{\ve}(u)f = \sum \lambda_i \overline{b_i} 
\label{eq:TactsonK}
\end{equation}
which reduces to the identity map, or complex conjugation, or quaternionic conjugation in $\BK$ depending on the value of $(p-q) \bmod 8$. Note that 
\begin{gather}
\{ 1, \be_\iu, \be_\ju, \be_\ku \} \stackrel{\overline{\phantom{\lambda}}}{\longmapsto} 
\{ 1, -\be_\iu, -\be_\ju, -\be_\ku \}
\label{eq:conjugateaction}
\end{gather}
due to part (ii) from Lemma~\ref{properties1}. The rest of our claim follows then easily from the $\BR$-linearity of $\ta{T}{\ve}$ and the fact that it is an anti-involution. Thus, we get 
$\overline{\lambda \mu} =  \overline{\mu}\overline{\lambda}$ for any $\lambda,\mu \in \BK.$

4. Part (d) of Theorem~\ref{th:structure} implies that in the monomial basis $\cb{B}$ of a simple Clifford algebra $\cl_{p,q}$ one can always find a (non-unique) set $\cb{M}$ of $N=2^k$ monomials 
$\{m_i\}_{i=1}^N$ for the chosen primitive idempotent $f$ such that the polynomials $\{ m_i f\}_{i=1}^N$  are linearly independent over $\BK$ and give a basis in $S=\cl_{p,q}f$ as a right $\BK$-vector space.\footnote{In a semisimple Clifford algebra, we have $N=2^{k-1}$. Semisimple algebras are considered in Section~\ref{action2}.} Thus, 
\begin{equation}
S = \cl_{p,q}f = \spn_\BK \{m_1f,m_2f,\ldots,m_N f\} \;\; \mbox{and} \;\; \psi = \sum_{i=1}^N m_i f \lambda_i,
\label{eq:m}
\end{equation}
where $ \lambda_i \in \BK,$ for any spinor $\psi \in S.$

The monomials $\cb{M}=\{m_i\}_{i=1}^N$ from~(\ref{eq:m}) have, as we will see shortly, a very interesting property: Let $f$ be a chosen primitive idempotent in $\cl_{p,q}$. For example, we will set $f$ to equal the polynomial~(\ref{eq:f}) where all signs are plus. Then,
\begin{equation}
\cb{F} = \{ m_1 f m_1^{-1}, m_2 f m_2^{-1},\ldots,  m_N f m_N^{-1}\}
\label{eq:F}
\end{equation}
provides a complete set of $N=2^k$ mutually annihilating primitive idempotents that add up to the unity $1$ of $\cl_{p,q}.$\footnote{We will see shortly that the monomial set $\cb{M}$ is a \textit{(left) transversal} for $G_{p,q}(f)$ in $G_{p,q}$.} 

Let $f_i =  m_i f m_i^{-1}$ and $\BK_i =f_i\cl_{p,q}f_i,$ $i=1,\ldots,N$.\footnote{From now on, we always sort the monomial set $\cb{M}$ by \texttt{InvLex} order so that $m_1=1$ and $f_1=f$.} For completeness, we recall how the monomials $\cb{M}$ act on the division rings $\BK_i$. Let $m_j \in \cb{M}$. Then, 
\begin{align}
\BK_1 = f_1 \cl_{p,q} f_1 \stackrel{m_j}{\rightarrow} 
m_j \BK_1 m_j^{-1} &= m_j (f_1 \cl_{p,q} f_1) m_j^{-1} \notag \\
                   &=(m_j f_1 m_j^{-1})(m_j \cl_{p,q} m_j^{-1})(m_j f_1 m_j^{-1}) \notag \\
                   &= f_j \cl_{p,q} f_j = \BK_j 
\end{align}
because the conjugate action on $\cl_{p,q}$ is an algebra automorphism. Thus, the division rings $\BK_i$ are all isomorphic. In addition, it can be verified that all rings $\BK_j$ as real subalgebras of 
$\cl_{p,q}$ have the same spanning set: $\{ 1 \}$, $\{1,\be_\iu \},$ or  $\{1,\be_\iu, \be_\ju, \be_\ku \}$ depending on the value of $(p-q) \bmod 8$ as shown  in~(\ref{eq:threecases}). Furthermore, the conjugate action of each such monomial  on $\BK$ modulo $f$ amounts to an algebra automorphism. This automorphism is the identity map when $\BK \cong \BR;$ it is either the identity map or ``complex'' conjugation when $\BK \cong \BC;$ and it is either the identity map or ``complex'' conjugation of two out of three subalgebras: $\spn_\BR \{1,\be_\iu \},$ $\spn_\BR \{1,\be_\ju\},$ and 
$\spn_\BR \{1,\be_\ku \},$ each isomorphic to $\BC,$ when $\BK \cong \BH.$ Hence, it is \emph{not the quaternionic conjugation} of $\BK$ realized by $\tp$ as shown in~(\ref{eq:conjugateaction}).  

Let $G_{p,q}$ be a finite group in any Clifford algebra $\cl_{p,q}$ (simple or semisimple) with a binary operation being just the Clifford product and defined as
\begin{equation}
G_{p,q} = \{ \pm \be_\iu  \; | \; \be_\iu \in \cb{B} \} \quad \mbox{with} \quad 
(\pm \be_\iu) (\pm \be_\ju) = \cml(\pm \be_\iu, \pm \be_\ju).
\label{eq:Gpq}
\end{equation}
This group of order $2 \cdot 2^{p+q} = 2^{n+1}$ is known as \emph{Salingaros vee group} and has been discussed, for example, in~\cite{salingaros1,salingaros2,salingaros3,varlamov}, but similar such groups
where also studied by Helmstetter~\cite{helmstetterCliffordGroups}. In particular,  $G_{p,q}$ is a discrete subgroup of $\Pin(p,q).$\footnote{See \cite[Sect. 17.2]{lounesto} for the definition of 
$\Pin(p,q)$. } We recall properties of this group next.

1. In a simple Clifford algebra, the group $G_{p,q}$ acts transitively via conjugation on any set $\cb{F}$ of primitive and orthogonal  idempotents. That is, let $f$ be any primitive idempotent and 
$g \in G_{p,q}.$ Then,
$$ 
g f g^{-1}  = g f \tp(g) 
$$
is another primitive idempotent. If  $g f g^{-1} \neq f$ then $g f g^{-1}$ annihilates $f$. The number of such idempotents in any complete set $\cb{F}$ is obviously $N = 2^k =\dim_\BK S$ where 
$k = q - r_{q-p}$ as before. Let $G_{p,q}(f)$ denote the stabilizer of $f$ under the conjugate action of $G_{p,q}$ and let $\cb{O}(f)$ be the orbit of $f$, then
$$
\hspace*{5ex} N=[G_{p,q}:G_{p,q}(f)]=|\cb{O}(f)|=|G_{p,q}|/|G_{p,q}(f)|=2\cdot 2^{p+q}/|G_{p,q}(f)|= 2^k
$$
where $[G_{p,q}:G_{p,q}(f)]$ is the index of the stabilizer $G_{p,q}(f)$ in $G_{p,q}$ and the bars $|\cdot|$ indicate the number of elements. This way, we can easily find the order of the stabilizer
$$
|G_{p,q}(f)| = 2 \cdot 2^{p+q}/2^k = 2^{1+p+q-k} = 2^{1+p+q-q+r_{q-p}} = 2^{1+p+r_{q-p}}.
$$
This property, thus, for all practical purposes, allows us to find a complete set of all mutually annihilating primitive idempotents in $\cl_{p,q}$
\begin{gather}
\cb{F} = \{f_1,f_2,\ldots,f_N\}
\label{eq:cbF}
\end{gather}
summing up to $1$ by setting, say, $f=f_1$ and then letting $G_{p,q}$ act on $f$ via the conjugation. Then, $\cb{F} = \cb{O}(f)$ since the action of $G_{p,q}$ is transitive. We summarize properties of the stabilizer in the following proposition and also in Tables~1, 2, and~3 in Appendix~\ref{AppendC}.
\begin{proposition}
Let $\cl_{p,q}$ be a simple Clifford algebra, $p-q \neq 1 \bmod 4$ and $p+q\leq 9$. Let $f$ be any primitive idempotent in the set  $\cb{F}$ and let $G_{p,q}(f)$ be its stabilizer in $G_{p,q}$ under the conjugate action. Then, 
\begin{itemize}
\item[(i)] $G_{p,q}(f) \lhd  G_{p,q}$, that is, $G_{p,q}(f)$ is a normal subgroup of $G_{p,q}$ and $|G_{p,q}(f)| = 2^{1+p+r_{q-p}}.$ 
\item[(ii)] $G_{p,q}(f)$ is a $2$-primary Abelian group when $p-q = 0,1,2 \bmod 8$ (real simple case) or $p-q = 3,7 \bmod 8$ (complex simple case), whereas $G_{p,q}(f)$ is a non Abelian $2$-group when 
$p-q = 4,5,6 \bmod 8$ (quaternionic simple case).
\item[(iii)] Let $k = q-r_{q-p}.$ Then, $G_{p,q}(f)$ is generated multiplicatively by $s$ non-unique elements
\begin{gather}
G_{p,q}(f) = \langle g_1,g_2, \ldots, g_s \rangle
\label{eq:gens1}
\end{gather}
where $s=k+1$ when $p-q = 0,1,2 \bmod 8$ or $p-q = 3,7 \bmod 8$, whereas $s=k+2$ when $p-q = 4,5,6 \bmod 8$. 
\item[(iv)] The orders of the generators $g_1,g_2,\ldots,g_s$ are $2$ or $4$, and the structure of the stabilizer group $G_{p,q}(f)$ is shown in Tables~1, 2, and 3 in Appendix~\ref{AppendC}.
\item[(v)] Let $m_j$ be any element in $G_{p,q}(f)$ and let $f$ have the form~(\ref{eq:f}). Then, the stability of the idempotent $m_j f m_j^{-1}=f$ implies 
$$
m_j \be_{\iu_1} m_j^{-1} = \be_{\iu_1}, \; 
m_j \be_{\iu_2} m_j^{-1} = \be_{\iu_2}, \; 
\ldots, \;
m_j \be_{\iu_k} m_j^{-1} = \be_{\iu_k}. 
$$
That is, the set of commuting monomials $\cb{T}= \{\be_{\iu_1},\ldots,\be_{\iu_k}\}$ in $f$ is point-wise stabilized by $G_{p,q}(f).$
\end{itemize}
\label{prop3} 
\end{proposition}
\noindent
All statements in the above proposition have been derived with {\CLIFFORD} for $\dim V \leq 9$. Property (v) came as rather unexpected: Instead of permuting the commuting monomials~$\cb{T}$ in the idempotent $f$ shown in ~(\ref{eq:f}), this set is stabilized point-wise by $G_{p,q}(f).$ Thus, all primitive idempotents in $\cb{F}$ have the same stabilizer group. We also comment that the direct product decomposition of $G_{p,q}(f)$ in the quaternionic case includes normal subgroups $F_3$ and $F_2$ of orders $16$ and~$8$, respectively. These groups contain the \textit{commutator subgroup} 
$G_{p,q}' = \{1,-1\}$ of $G_{p,q}.$\footnote{The \textit{commutator subgroup} $G'$ of a group $G$ is $G'=[G,G] = \langle [x,y] \,|\, x,y \in G \rangle$. That is, it is a subgroup of $G$ generated by
all commutators $[x,y]=xyx^{-1}y^{-1}$ for $x,y \in G$. In general, $G' \lhd G$ and $G/G'$ is Abelian~\cite[Prop. 5.57]{rotman}. We have $G_{p,q}'=\{1,-1\}$ since any two monomials in $G_{p,q}$ either commute or anti-commute.} Furthermore, $G_{p,q}(f)$ is solvable in all three cases since it is a finite $2$-group. As such, it does have a \textit{normal series} with Abelian factor groups. Thus, the \textit{derived series} of $G_{p,q}(f)$ will terminate with $(G_{p,q}(f))^{(n)}=\{1\}$ for some $n$ where $(G_{p,q}(f))^{(n)} = ((G_{p,q}(f))^{(n-1)})'$ is the commutator subgroup of
$(G_{p,q}(f))^{(n-1)}$~\cite{rotman}.      

2. There is a bijection $\theta$ (basic from the group theory) from the coset space $G_{p,q}/G_{p,q}(f)$ to the orbit $\cb{O}(f)$ which is given as 
\begin{gather}
\theta: G_{p,q}/G_{p,q}(f) \longrightarrow \cb{O}(f) \notag \\
 m_i G_{p,q}(f) \stackrel{\theta}{\mapsto} f_i = m_i f m_i^{-1} = m_i f \tp(m_i)
\label{eq:theta}
\end{gather}
where the set of coset representatives $\{ m_i\}_{i=1}^N$ is in fact the same as the set $\cb{M}$ of monomials~(\ref{eq:m}). That is, these coset representatives\footnote{The coset representatives are precomputed for all Clifford algebras $\cl_{p,q}, n=p+q \le 9$, in \CLIFFORD~\cite{ablamfauser2009}. They are also shown in \cite{ablamowicz1998}.} allow one to find a spinor basis $\cb{S}_j$ in $S_j = \cl_{p,q}f_j$ over $\BK = f\cl_{p,q}f$  by simply taking this set 
\begin{equation}
\cb{S}_j =\{m_1 f_j,\ldots,m_N f_j\}, \quad 1 \le j \le N,
\label{eq:spinorbasisj}
\end{equation}
as such basis.\footnote{Later, when we will be computing matrices of Clifford elements in spinor representation, we will always take the \textit{ordered} basis $\cb{S}_1 =[m_1 f_1,\ldots,m_N f_1]$.}
Of course, among these generators we have the  identity $1$, so that $1 f_j = f_j$ for every $j$ is the primitive idempotent whereas all other basis elements in $\cb{S}_j$ are nilpotent of index $2.$

Let $m_j \in \cb{M}$ and $f_j = m_j f m_j^{-1}$ where as before $f$ is a primitive idempotent defined as in (\ref{eq:f}) with all plus signs. Then, obviously,
\begin{itemize}
\item[(i)] $f_j^2 =(m_j f m_j^{-1})(m_j f m_j^{-1}) = m_j f m_j^{-1} = f_j$. 
\item[(ii)] $f_j$ is primitive as it has $k$ factors -- as many as $f$ does. Observe that
\begin{gather}
f_j = m_j f m_j^{-1} = 
\frac12(1+m_j \be_{\iu_1}m_j^{-1}) \cdots \frac12(1+m_j\be_{\iu_k}m_j^{-1})
\label{eq:actiononf}
\end{gather}
and, for all $1\le s,t \le k$ and $1 \le j \le N$ we have: 
\begin{enumerate}
\item $(m_j \be_{\iu_s}m_j^{-1})^2=1$ since $\be_{\iu_s}^2=1,$
\item $(m_j \be_{\iu_s}m_j^{-1})(m_j \be_{\iu_t}m_j^{-1}) = (m_j \be_{\iu_t}m_j^{-1})(m_j \be_{\iu_s}m_j^{-1})$ as $\be_{\iu_s}\be_{\iu_t} = \be_{\iu_t}\be_{\iu_s}$, 
\item $m_j \be_{\iu_s}m_j^{-1} \neq  m_j \be_{\iu_t}m_j^{-1}$ when $\be_{\iu_s} \neq \be_{\iu_t}$ since conjugation is a bijection. In fact, $m_j \be_{\iu_s}m_j^{-1}= \be_{\iu_s}$ for any 
$1 \leq s \leq k$ due to~(v) in Proposition~\ref{prop3}.
\end{enumerate}
\item[(iii)] Let $m_i \neq 1.$ Then, $(m_i f)(m_i f) =  \alpha_i (m_if m_i^{-1})f=\alpha_i f_if=0$ where $\alpha_i = m_i^2$ since $f_i \neq f$ and $f_i$ and $f$ are mutually annihilating. Thus, $m_if$ is nilpotent. Similarly, one shows that $m_i f_j$ is nilpotent as long as $m_i \neq 1.$
\end{itemize}

3. In addition to acting on the idempotent set $\cb{F}$ via the conjugation~(\ref{eq:theta}), the group $G_{p,q}$ has a representation on the coset space $G_{p,q}/G_{p,q}(f)$ since $G_{p,q}(f)$ has a finite index $N=2^{q-r_{q-p}}.$ Thus, according to the Representation on Cosets Theorem (see \cite[Theorem 2.88]{rotman}), there exists a homomorphism $\varphi: G_{p,q} \rightarrow S_N$ with $\ker \varphi \leq G_{p,q}(f)$. Here, by $S_N$ we denote the symmetric group on an $N$-element set. This action is important for our considerations since it amounts to permuting basis elements in each spinor space $S_j$. Thus, we look at it in greater detail.

For each monomial $m_j \in G_{p,q}$,\footnote{In fact, it is enough to pick $m_j$ from the left transversal $\cb{M}$.} define a left \textit{translation} 
\begin{gather}
\tau_{m_j}: G_{p,q}/G_{p,q}(f) \rightarrow G_{p,q}/G_{p,q}(f)
\label{eq:tau}
\end{gather}
by $\tau_{m_j}(m_i G_{p,q}(f)) = (m_j m_i) G_{p,q}(f)).$ Obviously, $\tau_{m_j}$ is a bijection and $\tau_{m_j} \in S_{G_{p,q}/G_{p,q}(f)}$, the symmetric group on the finite coset space. It is 
routine now to define 
$$
\varphi: G_{p,q} \rightarrow S_{G_{p,q}/G_{p,q}(f)}
$$ 
via  $\varphi(m_j) = \tau_{m_j}$ and show that $\varphi$ is a homomorphism and $\ker \varphi \leq G_{p,q}(f).$ Finally, $S_{G_{p,q}/G_{p,q}(f)} \cong S_N$ since $|G_{p,q}/G_{p,q}(f)|=N.$

Above we have discovered that the stabilizer group $G_{p,q}(f)$ is normal in $G_{p,q}$. Thus, the action of $G_{p,q}$ on the cosets via the left translation~(\ref{eq:tau}) really amounts to the coset multiplication since
$$
m_j (m_i G_{p,q}(f)) = (m_j G_{p,q}(f))(m_i G_{p,q}(f))
$$
in the quotient group $G_{p,q}/G_{p,q}(f)$.

We have already established that there is a bijection $\pi$ between the cosets in $G_{p,q}/G_{p,q}(f)$ and the spinor basis $\cb{S}$ in the spinor ideal $S = \cl_{p,q}f$:
\begin{gather}
m_i G_{p,q}(f) \stackrel{\pi}{\longmapsto} m_i f, 
\label{eq:pi}
\end{gather}
where $m_i \in \cb{M}$ give the spinor (ordered) basis $\cb{S} = [m_1 f, m_2 f, \ldots, m_N f]$. We will prove in Lemma~\ref{properties2} that $\pi$ is well-defined.

The left multiplication $\tau_{m_j}$ for each $m_j$ in $\cb{M}$ induces a bijection (permutation) $\kappa_{m_j}$ modulo the commutator group $G_{p,q}' =\{ 1,-1\}$ on the basis~$\cb{S}$ which also acts as left translation on these basis elements. We have the following commutative diagram:
\begin{equation}
\bfig
\square(0,0)|arra|/>`>`>`>/<1000,500>%
[G_{p,q}/G_{p,q}(f)`\cb{S}`G_{p,q}/G_{p,q}(f)`\cb{S};%
 \pi `\tau_{m_j}`\kappa_{m_j} =\,  \pi  \circ  \tau_{m_j} \circ  \pi^{-1}`\pi]
\efig
\label{eq:diagram5}
\end{equation}
where $\kappa_{m_j}: m_i f  \mapsto m_j  m_i f =  c_{j,i}^k  m_k  f$ for some real coefficients $c_{j,i}^k.$ In fact, since $\kappa_{m_j}$ is a bijection, for each pair $(j,i)$ there is exactly one value of the index $k$ such that $c_{j,i}^k = 1$ or $-1$, and it is zero for all other values of the index $k$. 
\begin{lemma}
\label{properties2}
Consider the spinor representation of $\cl_{p,q}$ in the spinor ideal $S_1$ with the ordered basis $\cb{S}_1 = [m_1 f_1,\ldots,m_N f_1]$ with $\alpha_i = m_i^2.$ Let $f \in \cb{F}$ be any primitive idempotent.
\begin{itemize}
\item[(a)] The mapping $\pi : G_{p,q}/G_{p,q}(f) \rightarrow \cb{S}$ where $m_i G_{p,q}(f) \stackrel{\pi}{\mapsto} m_i f$ is well-defined.
\item[(b)] Let $m_l f_k \in \cb{S}_k$ where $\cb{S}_k$ is a basis in the spinor ideal $S_k = \cl_{p,q}f_k$ and let $m_i f_1$ be any basis element in our chosen basis $\cb{S}_1$ in $S_1 = \cl_{p,q}f_1.$ Then, for any $m_l,m_i \in \cb{M}$,
\begin{equation}
(m_l f_k)(m_i f_1) = \sum_{j=1}^{N}c_{l,k}^{i,j} m_j f_1=\begin{cases} 
                       0                                  &\textit{$i \neq k$;}\\
                       m_l (m_k f_1) = c_{l,k}^j m_j f_1  &\textit{$i = k$,}
                       \end{cases}
\end{equation}
where, for any given pair $(l,k)$, there is exactly one value of the index~$j$ such that $c_{l,k}^j \stackrel{def}{=} c_{l,k}^{k,j} = 1$ or $-1$, and it is zero for all other values of~$j$. 
\item[(c)] Let $m_l,m_k \in \cb{M}$ and let $f_1 \in \cb{F}.$ Then, 
\begin{gather}
[m_k,m_l] = [m_k^{-1},m_l^{-1}] = \clkj{k}{l}{j}\clkj{l}{k}{j} \bmod f_1
\label{eq:comm}
\end{gather}
for some index $j.$
\item[(d)] The coefficients $\clkj{i}{j}{k}$ and $\clkj{i}{k}{j}$ satisfy the relation $\clkj{i}{j}{k} = \alpha_i \clkj{i}{k}{j}.$
\item[(e)]  Let $\psi_k = \sum_{i=1}^N m_i f_k \lambda_i$ be a spinor in the $k$-th spinor ideal $S_k = \cl_{p,q}f_k$ with components $\lambda_i \in \BK = f_1\cl_{p,q}f_1.$ Then, the matrix $[\psi_k]$ in the spinor representation has each $l$-th column zero when $l\neq k$ and in its $k$-th column, for each index $i$, in the $j$-th row it has exactly one (potentially) non-zero entry 
$\clkj{i}{k}{j} \lambda_{i,k}$ where $\lambda_{i,k} = m_k^{-1} \lambda_i m_k = m_k \lambda_i m_k^{-1}$.
\end{itemize}   
\end{lemma}
\begin{proof}
(a) Suppose that $\pi(m_i G_{p,q}(f)) = \pi(m_i' G_{p,q}(f))$ for some $m_i,m_i'$ in $G_{p,q}.$ Then, $m_i f = m_i' f$ or, $f = (m_i^{-1} m_i') f.$ Let $g = m_i^{-1} m_i'.$ Thus, $ f = gf$ and
$$
\tp(m_i f) = f m_i^{-1} = f (m_i')^{-1}, \quad \mbox{or,} \quad f(m_i^{-1}m_i') = fg = f
$$ 
because $\tp(f) = f$ for any primitive idempotent by part~(iii) of Lemma~\ref{properties1}, and $\tp(m_i) = m_i^{-1}$ for every $m_i \in G_{p,q}$ by part~(i) of the same lemma. Thus, $f = g f g^{-1}$ and $g \in G_{p,q}(f).$ Therefore, $\pi$ is well-defined.

(b) Observe that $(m_l f_k)(m_i f_1) = m_l f_k (m_i f_1 m_i^{-1})m_i = m_l f_k f_i m_i.$ The latter is zero when $i \neq k$ since $f_kf_i = 0$. When $i=k,$ we get
\begin{align}
(m_l f_k)(m_i f_1) 
&= m_l f_k f_i m_i = m_l f_k m_k = m_l m_k (m_k^{-1} f_k m_k) \notag\\
&=m_l (m_k f_1) = c_{l,k}^{j} m_j f_1.
\label{eq:clkj}
\end{align}

(c) Since $\cb{M} \subset G_{p,q},$ the monomials $m_k$ and $m_l$ either commute or anticommute. Thus, from~(\ref{eq:clkj}) we get $m_l (m_k f_1) = \clkj{l}{k}{j} m_j f_1$ and 
$m_k (m_l f_1) = \clkj{k}{l}{j} m_j f_1$ for some $m_j \in \cb{M}.$ Therefore, 
$$
\clkj{k}{l}{j}\clkj{l}{k}{j} f_1 = [m_k^{-1},m_l^{-1}] f_1 =  [m_k,m_l] f_1 
$$
where $[m_k,m_l] \in G_{p,q}'.$ 

(d) Since $m_i m_k f_1 = \clkj{i}{k}{j} m_j f_1$ and $m_i m_j f_1 = \clkj{i}{j}{k} m_k f_1,$ we get 
$$
m_k f_1 = \clkj{i}{k}{j} m_i^{-1} m_j f_1 = 
\clkj{i}{k}{j} \alpha_i m_i m_j f_1 =\clkj{i}{j}{k} m_i m_j f_1
$$
so $\clkj{i}{j}{k}= \alpha_i \clkj{i}{k}{j}.$ 

(e) To find entries in the $l$-th column of $[\psi_k]$ in the basis $\cb{S}_1,$ we need to compute the following product for each index $i$:
\begin{gather}
(m_i f_k \lambda_i) (m_l f_1) = m_i \lambda_i f_k m_l f_1 
                              = m_i \lambda_i (f_k m_l f_1 m_l^{-1}) m_l 
                              = m_i \lambda_i (f_k f_l) m_l
\label{eq:psik1}
\end{gather}
where $f_k f_l = 0$ unless $l = k$ since $f_k$ and $f_l$ are mutually annihilating otherwise. Thus, we continue under the assumption that $l=k:$
\begin{align}
(m_i f_k \lambda_i) (m_l f_1) 
&= m_i \lambda_i f_k m_k f_1 = m_i \lambda_i (m_k f_1 m_k^{-1}) m_k f_1 = m_i \lambda_i m_k f_1
\notag \\
&= m_i m_k (m_k^{-1} \lambda_i m_k) f_1 = m_i m_k f_1 \lambda_{i,k}
                                        = (m_j f_1) (\clkj{i}{k}{j} \lambda_{i,k})    
\end{align}
where $\lambda_{i,k} = m_k^{-1} \lambda_i m_k$. We have repeatedly used the fact that elements of the division ring $\BK$ commute with $f_1$ and that $G_{p,q}$ acts via conjugation on~$\BK.$ Thus, this last equality tells us that, for each index $i$, the entry $\clkj{i}{k}{j} \lambda_{i,k}$ is located in the $j$-th row of the $k$-th column in $[\psi_k]$. 
\end{proof} 

In order to clarify part (e) of the last lemma, we recall that the left multiplication maps $\kappa_{m_i}$ on $\cb{S}_1$ are bijections, that is, the triples $(i,j,k)$ are uniquely determined in the following sense: For each pair of any two indices, the third index is uniquely determined whenever the coefficients $\clkj{i}{j}{k}$ and $\clkj{i}{k}{j}$ are both non-zero.

We illustrate the above lemma with the following examples.

\begin{example}
Consider $\cl_{2,2} \cong \Mat(4,\BR)$ with $f = \frac14(1 + \be_{13})(1 + \be_{24})$. The  monomial list $\cb{M}= [1,\be_{1},\be_{2},\be_{12}]$ contains our chosen coset representatives in 
$G_{2,2}/G_{2,2}(f)$ and the set $\cb{F}$ contains these four idempotents:
\begin{gather}
f_{1} = 1\, f\, 1^{-1}, \qquad f_{2}= \be_{1} f \be_{1}^{-1}, 
\qquad 
f_{3} = \be_{2} f \be_{2}^{-1},\qquad  f_{4} = \be_{12} f \be_{12}^{-1}.
\label{eq:idempts6}
\end{gather}
Then, the list $\cb{S}_1 = [f_1,\be_{1} f_1,\be_{2} f_1,\be_{12} f_1]$ is the ordered basis in $S_1 = \cl_{2,2}f_1$ and similarly for the other three spinor ideals. There are, as expected, exactly sixteen non-zero coefficients $\clkj{l}{k}{j}$ which satisfy relation~(\ref{eq:clkj}). To save space, we display them in the following matrix:
\begin{gather}
C = 
\left[\begin{matrix} \clkj{1}{1}{1} & \clkj{2}{2}{1} & \clkj{3}{3}{1} & \clkj{4}{4}{1}\\[1ex]
                     \clkj{2}{1}{2} & \clkj{1}{2}{2} & \clkj{4}{3}{2} & \clkj{3}{4}{2}\\[1ex]
                     \clkj{3}{1}{3} & \clkj{4}{2}{3} & \clkj{1}{3}{3} & \clkj{2}{4}{3}\\[1ex]
                     \clkj{4}{1}{4} & \clkj{3}{2}{4} & \clkj{2}{3}{4} & \clkj{1}{4}{4}
\end{matrix}\right]
=
\left[\begin{matrix} 1 & \phm 1 & 1 & -1\\[1ex]
                     1 & \phm 1 & 1 & -1\\[1ex]
                     1 & -1 & 1 &  \phm 1\\[1ex]
                     1 & -1 & 1 &  \phm 1\end{matrix}\right]
\label{eq:CExample6}
\end{gather}
where $C_{j,k} = \clkj{l}{k}{j}$ for some index $l.$  From this matrix it is easy to read off matrices that represent each of the sixteen basis elements in the four minimal ideals: Matrix $[m_l f_k]$ of the basis element $m_l f_k \in \cb{S}_k$ in the spinor representation of $\cl_{2,2}$ in~$S_1$ has only one non-zero entry in its $j$-th row and $k$-th column that equals~$C_{j,k}.$

Let $\psi_k = \sum_{i=1}^4 m_i f_k \lambda_i \in S_k$ where $\lambda_i = \psi_{ik} \in \BK = f_k \cl_{2,2} f_k \cong \BR$; the monomials $m_i \in \cb{M}$ and the idempotents are shown 
in~(\ref{eq:idempts6}). Then, since the group $G_{2,2}$ acts trivially on $\BK$, we have $\lambda_{i,k} = m_k^{-1} \lambda_i m_k = \lambda_i,\forall i,k,$ and we find:
\begin{gather}
[\Psi]=[\psi_1] + [\psi_2] + [\psi_3] + [\psi_4] = 
\left[\begin{matrix} \psi_{11} & \phm \psi_{21} & \psi_{31} &     -\psi_{41} \\
                     \psi_{21} & \phm \psi_{11} & \psi_{41} &     -\psi_{31} \\
                     \psi_{31} &     -\psi_{41} & \psi_{11} & \phm \psi_{21} \\
                     \psi_{41} &     -\psi_{31} & \psi_{21} & \phm \psi_{11} \end{matrix}\right]
\label{eq:sumofspinors1} 
\end{gather}
The display~(\ref{eq:sumofspinors1}) clearly shows, that (i) the sign distribution matches that of the matrix~$C$ in~(\ref{eq:CExample6}); (ii) entries in columns two, three, and four are essentially, up to the sign, permutations of the entries in the first column; (iii) the $(j,k)$ entry of $[\Psi]$ is exactly $\clkj{i}{k}{j} \lambda_{i,k}$ as predicted by part~(e) of Lemma~\ref{properties2}.
\end{example}

\begin{example}
Consider $\cl_{3,0} \cong \Mat(2,\BC)$ with $f = \frac12(1 + \be_1)$. The monomial list $\cb{M}= [1,\be_2]$ shows our chosen coset representatives in $G_{3,0}/G_{3,0}(f)$ and the set $\cb{F}$ contains these two idempotents:
\begin{gather}
f_{1} = 1 f 1^{-1} = \frac12(1 + \be_1), \quad
f_{2} = \be_{2} f \be_{2}^{-1} = \frac12(1 - \be_1). 
\label{eq:idempts7}
\end{gather}
Then, the list $\cb{S}_1 = [f_1,\be_2 f_1]$ is the ordered basis in $S_1 = \cl_{3,0}f_1$ and the set $\cb{S}_2 = \{ f_2, \be_2 f_2\}$ is a basis in $S_2 = \cl_{3,0}f_2.$ There are, as expected, exactly four non-zero coefficients $\clkj{1}{1}{1}=\clkj{1}{2}{2}=\clkj{2}{1}{2}=\clkj{2}{2}{1}=1$ which satisfy relation~(\ref{eq:clkj}). Matrices of the four basis elements from $\cb{S}_1$ and $\cb{S}_2$ in the representation of $\cl_{3,0}$ in $S_1$, are as follows:
\begin{align}
[f_1] &= \left[\begin{matrix} \clkj{1}{1}{1} & 0\\[0.5ex] 0 & 0\end{matrix}\right]=
         \left[\begin{matrix} 1 & 0\\[0.5ex] 0 & 0\end{matrix}\right], &
[\be_{2} f_2] &= \left[\begin{matrix} 0 & \clkj{2}{2}{1}\\[0.5ex] 0 & 0\end{matrix}\right]=
                 \left[\begin{matrix} 0 & 1\\[0.5ex] 0 & 0\end{matrix}\right],\notag \\[0.5ex]
[\be_{2} f_1] &= \left[\begin{matrix} 0 & 0\\[0.5ex] \clkj{2}{1}{2} & 0\end{matrix}\right]=
                 \left[\begin{matrix} 0 & 0\\[0.5ex] 1 & 0\end{matrix}\right], &
[f_2] &= \left[\begin{matrix} 0 & 0\\[0.5ex] 0 & \clkj{1}{2}{2}\end{matrix}\right]=
         \left[\begin{matrix} 0 & 0\\[0.5ex] 0 & 1\end{matrix}\right].
 \end{align}
Let $\psi_k = \sum_{i=1}^2 m_i f_k \lambda_i \in S_k$ where $\lambda_i =\psi_{i1}+\psi_{i2} \be_{23} \in \BK=f_1 \cl_{3,0} f_1 \cong \BC$; monomials $m_i \in \cb{M}$  and the idempotents are shown 
in~(\ref{eq:idempts7}). Then, since this time the group $G_{3,0}$ acts non trivially on $\BK$ via conjugation, we have
$$
\psi_{i1}+\psi_{i2} \be_{23} \stackrel{m_1}{\longmapsto} \psi_{i1}+\psi_{i2} \be_{23}, \quad \psi_{i1}+\psi_{i2} \be_{23} \stackrel{m_2}{\longmapsto} \psi_{i1}-\psi_{i2} \be_{23},
$$
and we find:
\begin{gather}
[\Psi]=[\psi_1] + [\psi_2] = 
\left[\begin{matrix} \psi_{11}+\psi_{12}\be_{23} & \psi_{21} - \psi_{22}\be_{23}\\
                     \psi_{21}+\psi_{22}\be_{23} & \psi_{11} - \psi_{12}\be_{23}
\end{matrix}\right]
\label{eq:sumofspinors2}
\end{gather}
again in agreement with part~(e) of Lemma~\ref{properties2}.
\end{example}

Combining now the conjugate action of $G_{p,q}$ on $\cb{F}$ with the permutations modulo $G_{p,q}'$ on our chosen spinor basis $\cb{S} = \cb{S}_1,$ we get the following proposition.

\begin{proposition}
The group $G_{p,q}$ permutes the basis elements in $\cb{S}$ modulo the commutator group $G_{p,q}'=\{1,-1\}$. That is, it acts on $\cb{S}$ via the left translation $\kappa_{m_j}$ modulo $\{1,-1\}$ where 
$m_j \in \cb{M}$ is the coset representative of $m_j G_{p,q}(f)$ in the quotient group $G_{p,q}/G_{p,q}(f).$
\label{prop4}
\end{proposition}
The above proposition is of course true if we replace the idempotent $f$ with any conjugate to it idempotent $f_k$. Now we are ready to prove the following results. 
\begin{proposition}
Let $\psi,\phi \in S = \cl_{p,q}f.$ Then, $\ta{T}{\ve}(\psi)\phi \in \BK.$ In particular, 
$\ta{T}{\ve}(\psi)\psi \in \BR f \subset \BK.$
\label{prop5}
\end{proposition}
\begin{proof}
Let $\psi = \sum_i m_i f \lambda_i$ and $\phi = \sum_j m_j f \mu_j,$ $\lambda_i,\mu_j \in \BK,$ be two spinors in~$S$. Recall that $\ta{T}{\ve}(\lambda_i) = \overline{\lambda_i}$ from~(\ref{eq:TactsonK}) while $\ta{T}{\ve}(m_i) = m_i^{-1}$ and $\ta{T}{\ve}(f)=f$ from Lemma~\ref{properties1}. Then we have:
\begin{align}
\ta{T}{\ve}(\psi)\phi 
&= \ta{T}{\ve}\Bigl(\sum_i m_i f \lambda_i\Bigr)\Bigl(\sum_j m_j f \mu_j\Bigr) 
 = \Bigl(\sum_i \overline{\lambda_i} f \ta{T}{\ve}(m_i)\Bigr)\Bigl(\sum_j m_j f \mu_j\Bigr) \notag \\
&=\sum_i \overline{\lambda_i} f \underbrace{\ta{T}{\ve}(m_i)m_i}_{1} f \mu_i +
  \sum_i\sum_{\substack{j\\ j\neq i}} \overline{\lambda_i}f\ta{T}{\ve}(m_i)m_jf \mu_j
\notag \\
&=\sum_i \overline{\lambda_i}\mu_i f +
  \sum_i\sum_{\substack{j\\ j\neq i}} \overline{\lambda_i} m_i^{-1}m_i f\ta{T}{\ve}(m_i)m_jf 
m_j^{-1}m_j \mu_j
\notag \\
&=\sum_i \overline{\lambda_i}\mu_i f +
  \sum_i\sum_{\substack{j\\ j\neq i}} \overline{\lambda_i}m_i^{-1} 
        \underbrace{(m_if m_i^{-1})}_{f_i}\underbrace{(m_jf m_j^{-1})}_{f_j}m_j \mu_j
\notag \\
&=\sum_i \overline{\lambda_i}\mu_i f \in \BK f = f \BK = \BK 
\end{align}
because $f_if_j = f_jf_i =0$ whenever $i \neq j$, that is, $f_i$ and $f_j$ are two mutually annihilating idempotents. 
\end{proof} 
\begin{corollary}
Let $\psi \in S = \cl_{p,q}f.$ Then, 
$\ta{T}{\ve}(\psi)\psi \in \BR \cong \BR f \subset \BK.$
\end{corollary}

\begin{definition}
Let $G_{p,q}^\ve = \{ g \in \cl_{p,q}\; |\; \ta{T}{\ve}(g)g = 1\}$.
\end{definition}

\begin{corollary}
\hspace*{0.1ex}
\begin{itemize}
\item[(i)] The $\BK$-valued inner product $S_j \times S_j \rightarrow \BK$ defined as
\begin{equation}
(\psi,\phi) \mapsto \tp(\psi)\phi = \lambda f_j = f_j \lambda, \quad \lambda \in \BK
\label{eq:psiphi}
\end{equation}
is invariant under the group $G_{p,q}^\ve.$\footnote{In~\cite{part3} we provide a complete classification of all groups $G_{p,q}^\ve.$}
\item[(ii)] $G_{p,q}(f) \unlhd G_{p,q} \leq G_{p,q}^\ve.$
\end{itemize}
\end{corollary}

The above comments are very helpful in proving that indeed $\tp$ is a conjugation on $S$ and, therefore, the matrix of $\tp(u)$ is the Hermitian conjugate of the matrix of $u$ in the spinor representation of $\cl_{p,q}$ in $S = \cl_{p,q}f$ for any $u \in \cl_{p,q}.$ 
\begin{proposition}
Let $\cl_{p,q}$ be a simple Clifford algebra, $p-q \neq 1 \bmod 4$ and $p+q\leq 9$. Let $\psi_k \in S_k = \cl_{p,q}f_k$ and let $[\psi_k]$ (resp. $[\tp(\psi_k)]$) be the matrix of $\psi_k$ (resp. 
$\tp(\psi_k)$) in the spinor representation with respect to the ordered basis $\cb{S}_1 = [m_1 f_1,\ldots,m_N f_1]$ with $\alpha_i=m_i^2$. Then, 
\begin{equation}
[\tp(\psi_k)] = \begin{cases} 
    [\psi_k]^T         & \textit{if $p-q =0,1,2 \bmod 8;$} \\
    [\psi_k]^\dagger  & \textit{if $p-q =3,7 \bmod 8;$} \\
    [\psi_k]^\ddagger & \textit{if $p-q =4,5,6 \bmod 8;$}
\end{cases}
\label{eq:daggers}
\end{equation}
where $T$ denotes transposition, $\dagger$ denotes Hermitian complex  conjugation, and $\ddagger$ denotes Hermitian quaternionic conjugation.
\label{prop6}
\end{proposition}
\begin{proof}
Let $\psi_k = \sum_{i=1}^{N} m_i f_k \lambda_i$ be a spinor in $S_k=\cl_{p,q}f_k$ where $\lambda_i \in \BK = f_1 \cl_{p,q} f_1.$ In part (e) of Lemma~\ref{properties2} we have shown that the $k$-th column of $[\psi_k]$ is the only potentially non-zero column with the $(j,k)$ entry equal to $\clkj{i}{k}{j} \lambda_{i,k}$ where $\lambda_{i,k} = m_k \lambda_i m_k^{-1}$. Thus, it is enough to show that the $(k,j)$ entry of $[\tp(\psi_k)]$ is $\clkj{i}{k}{j} \overline{\lambda_{i,k}}$ where by $\overline{\lambda_{i,k}}$ we mean either the identity involution, complex conjugation, or quaternionic conjugation depending on the value of $(p-q) \bmod 8$.

In order to find entries in the $l$-th column of $[\tp(\psi_k)]$ in the basis $\cb{S}_1$, we proceed like in~(\ref{eq:psik1}):
\begin{align}
\tp(m_i f_k \lambda_i)(m_l f_1)
&= \tp(\lambda_i) f_k \tp(m_i) (m_l f_1) \notag \\
&= \overline{\lambda_i} f_k m_i^{-1} m_l f_1 = \overline{\lambda_i}\alpha_i (f_k m_i m_l f_1)     
\label{eq:tppsik1}
\end{align}
It should be clear, using the same argument as in the proof of part (e) of Lemma~\ref{properties2}, that the product $f_k m_i m_l f_1 = 0$ unless $m_i m_l \propto m_k$. Since $(i,j,k)$ is a triple such that $m_i m_j f_1= \clkj{i}{j}{k} m_k f_1,$ this means that in order for the product $f_k m_i m_l f_1$ not to be zero, we must have $l=j$. This is because the translation $\kappa_{m_i}$ is a bijection on $\cb{S}_1$. Therefore, we continue under the assumption that $l=j$:
\begin{align}
\tp(m_i f_k \lambda_i)(m_l f_1)
&= \overline{\lambda_i} \alpha_i f_k (m_i m_j f_1) 
 = \overline{\lambda_i} \alpha_i f_k \clkj{i}{j}{k} m_k f_1  \notag \\
&= \overline{\lambda_i} \alpha_i (m_k f_1 m_k^{-1})(m_k  \clkj{i}{j}{k} f_1)\notag \\
&= \overline{\lambda_i} \alpha_i m_k f_1 \clkj{i}{j}{k} 
 = m_k(m_k^{-1} \overline{\lambda_i} m_k) \alpha_i f_1 (\alpha_i \clkj{i}{k}{j})\notag \\
&= (m_k f_1) (m_k^{-1} \overline{\lambda_i} m_k) \clkj{i}{k}{j} 
 = (m_k f_1) (\clkj{i}{k}{j} \overline{\lambda_{i,k}}) 
\label{eq:tppsik2}
\end{align}
where we have used the identity $\clkj{i}{j}{k}= \alpha_i \clkj{i}{k}{j}$ from Lemma~\ref{properties2} part (d); the fact that $\alpha_i^2=1;$ and the following:
\begin{align}
m_k^{-1} \overline{\lambda_i} m_k 
&= m_k^{-1} \tp(\lambda_i) m_k = \tp(m_k) \tp(\lambda_i) \tp(m_k^{-1}) \notag \\
&=\tp(m_k^{-1} \lambda_i m_k) = \tp(\lambda_{i,k}) = \overline{\lambda_{i,k}}.
\end{align}
Thus, (\ref{eq:tppsik2}) shows that the $(k,j)$ entry of $[\tp(\psi_k)]$ is $\clkj{i}{k}{j} \overline{\lambda_{i,k}}.$ 
\end{proof}
We illustrate this proposition with the following examples.
\begin{example}
Consider $\cl_{2,2} \cong \Mat(4,\BR)$ with $f = \frac14(1+\be_{13})(1+\be_{24}).$ Then, $\BK=f\cl_{2,2}f = \BR f \cong \BR$ and $\cb{M} = \{1,\be_1,\be_2,\be_{12}\}$.  
Thus, any spinor in $S = \cl_{2,2}f$ has the form 
\begin{equation}
\psi = f\psi_1 + \be_1 f\psi_2 + \be_2 f\psi_3 + \be_{12} f\psi_4 
\label{eq:psij}
\end{equation} 
where $\psi_i \in \BK.$ Then, the matrices $[\psi]$ and $[\tp(\psi)]$ in the spinor representation are related via the matrix transposition:
\begin{equation}
[\psi] = \left[\begin{matrix} 
\psi_1 & 0 & 0 & 0\\
\psi_2 & 0 & 0 & 0\\
\psi_3 & 0 & 0 & 0\\
\psi_4 & 0 & 0 & 0
\end{matrix}\right], 
\qquad 
[\tp(\psi)] = \left[\begin{matrix} 
\psi_1 & \psi_2 & \psi_3 & \psi_4\\
0 & 0 & 0 & 0\\
0 & 0 & 0 & 0\\
0 & 0 & 0 & 0
\end{matrix}\right] 
\label{eq:psireal}
\end{equation}
as can be shown by direct computation. Notice that we can obtain the set $\cb{F}$ by acting on $f$ via conjugation with monomials from $\cb{M}$:
\begin{align*}
f_1 &= 1 f 1^{-1}               = \frac14(1+\be_{13})(1+\be_{24}), 
&
f_2 &= \be_1 f \be_1^{-1}       = \frac14(1-\be_{13})(1+\be_{24}),\\
f_3 &= \be_2 f \be_2^{-1}       = \frac14(1+\be_{13})(1-\be_{24}),
&
f_4 &= \be_{12} f \be_{12}^{-1} = \frac14(1-\be_{13})(1-\be_{24}).
\end{align*}
Thus, we have the decomposition of $\cl_{2,2}$ into a direct sum of $(\cl_{2,2},\BR)$-bimodules:
$$
\cl_{2,2} = \cl_{2,2}f_1 \oplus  \cl_{2,2}f_2 \oplus  \cl_{2,2}f_3 \oplus  \cl_{2,2}f_4  
$$
If instead of $f = f_1$ in (\ref{eq:psij}) we take $f_2,$ $f_3,$ or $f_4$ with the same set $\cb{M},$ we get, correspondingly, the second, the third, and fourth column  in the matrix~$[\psi]$. For example, for $\psi$ in $S_2 = \cl_{2,2}f_2$ we get:
\begin{equation}
[\psi] = \left[\begin{matrix} 
0 & \psi_2 & 0 & 0 & 0\\
0 & \psi_1 & 0 & 0 & 0\\
0 & -\psi_4 & 0 & 0 & 0\\
0 & -\psi_3 & 0 & 0 & 0
\end{matrix}\right], 
\qquad 
[\tp(\psi)] = \left[\begin{matrix} 
0 & 0 & 0 & 0\\
\psi_2 & \psi_1 & -\psi_4 & -\psi_3\\
0 & 0 & 0 & 0\\
0 & 0 & 0 & 0
\end{matrix}\right] 
\label{eq:psireal2}
\end{equation}
again related via the transposition. Finally, we consider the inner product on 
$$
S \times S \rightarrow \BK, \quad (\psi,\phi) \mapsto \tp(\psi)\phi.
$$
Let $\phi$ be another spinor in $S=\cl_{2,2}f$ expressed in a similar manner as in (\ref{eq:psij}). Then,
$$
\tp(\psi)\phi = (\psi_1 \phi_1 + \psi_2 \phi_2 + \psi_3 \phi_3 + \psi_4 \phi_4)f \in \BR f.
$$
We remark here, that the above inner product is neither Lounesto's $\beta_{+}$ nor $\beta_{-}$ inner products on $S$, hence it is different from `spinor metrics' commonly used, see~\cite{lounesto}:
\begin{align}
\beta_{+}(\psi,\phi) &= s_1 \tilde{\psi}\phi = 
         (-\psi_1 \phi_4+\psi_3 \phi_2+\psi_4 \phi_1-\psi_2 \phi_3)f \in \BR f\\
\beta_{-}(\psi,\phi) &= s_2 \bar{\psi}\phi = 
         (-\psi_1 \phi_4-\psi_3 \phi_2+\psi_4 \phi_1+\psi_2 \phi_3)f \in \BR f 
\end{align}
where $s_1 = s_2 = \be_{12}$ is a pure spinor. Here, the tilde in $\tilde{\psi}$ denotes reversion whereas the bar in $\bar{\psi}$ denotes Clifford conjugation. In signature $(2,2)$, the two forms 
$\beta_{+}$ and $\beta_{-}$ are invariant under $Sp(4,\BR).$ See~\cite[Table 1 and 2, p. 236]{lounesto}, while the bilinear form $\tp(\psi)\phi$ is invariant under $O(4,\BR)$.
\end{example}

\begin{example}
Consider $\cl_{3,0} \cong \Mat(2,\BC)$ with $f = \frac12(1+\be_1).$ Then, $\BK = f \cl_{3,0}f = \spn_\BR \{f,\be_{23}f \} \cong \BC$ and $\cb{M}= \{1,\be_2\}$. Thus, any spinor in $S = \cl_{3,0}f$ has the form
\begin{gather}
\psi = f \psi_1 + \be_2 f \psi_2 
     = f (\psi_{11} + \psi_{12}\be_{23}) + \be_2 f (\psi_{21} + \psi_{22}\be_{23})
\label{eq:psij2}
\end{gather}
where $\psi_i=\psi_{i1} + \psi_{i2}\be_{23} \in \BK$ and 
$\psi_{i1},\psi_{i2} \in \BR.$ Then, the matrices $[\psi]$ and $[\tp(\psi)]$ in the spinor representation are related via Hermitian complex  conjugation:
\begin{equation}
[\psi] = \left[\begin{matrix} 
\psi_{11}+\psi_{12} \be_{23} & 0\\
\psi_{21}+\psi_{22} \be_{23} & 0
\end{matrix}\right], 
\qquad 
[\tp(\psi)] = \left[\begin{matrix} 
\psi_{11} - \psi_{12}\be_{23} & \psi_{21}-\psi_{22} \be_{23}\\
0 & 0
\end{matrix}\right] 
\label{eq:psicomplex}
\end{equation}
as the direct computation shows. Like in the real case, the set $\cb{F}$ can be found by acting on $f$ via conjugation with monomials from $\cb{M}$: 
\begin{gather*}
f_1 = 1 f 1^{-1} = \frac12(1+\be_{1}), \qquad f_2 = \be_2 f \be_2^{-1} = \frac12(1-\be_{1}).
\end{gather*}

Thus, we have the decomposition of $\cl_{3,0}$ into a direct sum of $(\cl_{3,0},\BC)$-bimodules:
$$
\cl_{3,0} = \cl_{3,0}f_1 \oplus  \cl_{3,0}f_2.  
$$
If instead of $f = f_1$ in (\ref{eq:psij2}) we take $f_2$ with the same set $\cb{M},$ we get the second column in the matrix~$[\psi]$. For example, for $\psi$ in $S_2 = \cl_{3,0}f_2$ we get:
\begin{equation}
[\psi] = \left[\begin{matrix} 
0 & \psi_{21} - \psi_{22}\be_{23}\\
0 & \psi_{11} - \psi_{12} \be_{23}
\end{matrix}\right], 
\qquad 
[\tp(\psi)] = \left[\begin{matrix} 
0 & 0\\
\psi_{21} + \psi_{22} \be_{23} & \psi_{11}-\psi_{12} \be_{23}
\end{matrix}\right] 
\label{eq:complex2}
\end{equation}
which again are related via Hermitian complex conjugation. Finally, we consider the inner product on 
$$
S \times S \rightarrow \BK, \quad (\psi,\phi) \mapsto \tp(\psi)\phi.
$$
Let $\phi$ be another spinor in $S=\cl_{3,0}f$ expressed in a similar manner as in (\ref{eq:psij2}). Then,
\begin{multline*}
\tp(\psi)\phi = (\psi_{11} \phi_{11}+\psi_{22} \phi_{22}+\psi_{21} \phi_{21}+\psi_{12}\phi_{12}) +\\
(-\psi_{22} \phi_{21}-\psi_{12} \phi_{11}+\psi_{21} \phi_{22}+\psi_{11}\phi_{12}) \be_{23} \in \BK
\end{multline*}
Comparison with Lounesto's $\beta_{+}$ and $\beta_{-}$ inner products on $S$ shows that our product coincides with $\beta_{+},$ so it is invariant under $U(2)$.\footnote{We remark here that the product $\tp(\psi)\phi$ always coincides with $\beta_{+}$ in Euclidean signatures $(p,0)$ and with $\beta_{-}$ in anti-Euclidean signatures $(0,q)$.} Therefore, it is different from $\beta_{-}$ which invariant under $Sp(2,\BC):$
\begin{multline}
\beta_{-}(\psi,\phi) = s_2 \bar{\psi}\phi =(\psi_{22} \phi_{12}-\psi_{21} \phi_{11}-\psi_{12} \phi_{22}+\psi_{11} \phi_{21})+\\
(-\psi_{22} \phi_{11}-\psi_{21} \phi_{12}+\psi_{12} \phi_{21}+\psi_{11} \phi_{22})\be_{23} \in \BK 
\end{multline}
where $s_2 = \be_{2}$ are pure spinors.
\end{example}

\begin{example}
Consider $\cl_{2,4} \cong \Mat(4,\BH)$ with $f = \frac14(1+\be_{15})(1+\be_{26})$. Then, $\BK = f\cl_{2,4}f = \spn_\BR \{1,\be_3,\be_4,\be_{34}\} \cong \BH,$ $\cb{M}=\{1,\be_1,\be_2,\be_{12}\},$ and 
\begin{equation}
S = \cl_{2,4}f = \spn_\BK \{f,\be_1 f,\be_2 f,\be_{12}f \}.
\label{eq:Sj3}
\end{equation}
Let
\begin{equation}
\psi = f \psi_1 + \be_1 f \psi_2 + \be_2 f \psi_3 + \be_{12}f \psi_4 \in S
\label{eq:psi3}
\end{equation}
where $\psi_1,\psi_2,\psi_3,\psi_4 \in \BK$. Then, the matrices $[\psi]$ and $[\tp(\psi)]$ in the spinor representation are related via Hermitian quaternionic conjugation:
\begin{equation}
[\psi] = \left[\begin{matrix} 
\psi_{1} & 0 & 0 & 0\\
\psi_{2} & 0 & 0 & 0\\
\psi_{3} & 0 & 0 & 0\\
\psi_{4} & 0 & 0 & 0
\end{matrix}\right],\qquad
[\tp(\psi)] = \left[\begin{matrix} 
\bar{\psi}_1 & \bar{\psi}_2 & \bar{\psi}_3 & \bar{\psi}_4\\
0 & 0 & 0 & 0\\
0 & 0 & 0 & 0\\
0 & 0 & 0 & 0
\end{matrix}\right] 
\label{eq:psiquat}
\end{equation}
where $\bar{\psi}_i$ is the quaternionic conjugate. We compute the set $\cb{F}$:
\begin{align*}
f_1 &= 1 f 1^{-1}               = \frac14(1-\be_{15})(1+\be_{26}),
&
f_2 &= \be_1 f \be_1^{-1}       = \frac14(1-\be_{15})(1+\be_{26}),\notag \\
f_3 &= \be_2 f \be_2^{-1}       = \frac14(1+\be_{15})(1-\be_{26}),
&\
f_4 &= \be_{12} f \be_{12}^{-1} = \frac14(1-\be_{15})(1-\be_{26}).
\end{align*}
Thus, we have the decomposition of $\cl_{2,4}$ into a direct sum of $(\cl_{2,4},\BH)$-bimodules:
$$
\cl_{2,4} = \cl_{2,4}f_1 \oplus  \cl_{2,4}f_2 \oplus  \cl_{2,4}f_3 \oplus  \cl_{2,4}f_4  
$$
It can be again easily shown that if we take in (\ref{eq:Sj3}) and (\ref{eq:psi3}) idempotent $f_2,$ $f_3,$ or $f_4$ instead of $f$, then the matrix $[\psi]$ will have the second, the third or the fourth column nonzero, whereas the matrix $[\tp(\psi)]$ will be the Hermitian quaternionic conjugate of $[\psi]$. We won't display here the inner product $\tp(\psi)\phi$ except we remark that it is again different from $\beta_{+}$ and $\beta_{-}.$
\end{example}

Using now Proposition~\ref{prop6} as well as the fact that the (simple) Clifford algebra $\cl_{p,q}$ is decomposable into a direct sum of $N=2^{q-r_{q-p}}$ spinor $(\cl_{p,q},\BK)$-bimodules 
\begin{gather}
\cl_{p,q} = \cl_{p,q}f_1 \oplus \cdots \oplus \cl_{p,q}f_N,
\label{eq:decomp2a}
\end{gather}
we obtain the next result.\footnote{Decomposition~(\ref{eq:decomp2a}) is also valid in semi-simple Clifford algebras when $p-q = 1 \bmod 4$.}
\begin{proposition}
Let $\cl_{p,q}$ be a simple Clifford algebra, $p-q \neq 1 \bmod 4$ and $p+q\leq 9$. Let $u \in \cl_{p,q}$. Suppose that $[u]$ (resp. $[\tp(u)]$) is a matrix of $u$ (resp. $\tp(u)$) in the spinor representation with respect to the ordered basis $\cb{S}_1 = [m_1 f_1,\ldots, m_N f_N]$. Then, 
\begin{equation}
[\tp(u)] = \begin{cases} 
    [u]^T         & \textit{if $p-q =0,1,2 \bmod 8;$} \\
    [u]^\dagger  & \textit{if $p-q =3,7 \bmod 8;$} \\
    [u]^\ddagger & \textit{if $p-q =4,5,6 \bmod 8;$}
\end{cases}
\label{eq:matdaggers}
\end{equation}
where $T$ denotes transposition, $\dagger$ denotes Hermitian complex conjugation, and $\ddagger$ denotes Hermitian quaternionic conjugation.
\label{prop7}
\end{proposition}
\begin{proof}
Due to the direct sum decomposition (\ref{eq:decomp2a}), every element $u \in \cl_{p,q}$ can be written as a sum of unique spinors $\psi_k$ where $\psi_k \in \cl_{p,q}f_k.$ Since spinor representation is $\BR$-linear, we have that
\begin{gather}
[u] = [\psi_1] + \cdots + [\psi_N]
\label{eq:udecomp}
\end{gather}
where $\psi_k = u f_k.$ Furthermore, due to $\BR$-linearity of the involution $\tp$, we have
\begin{gather}
[\tp(u)] = [\tp(\psi_1)] + \cdots + [\tp(\psi_N)]. 
\label{eq:tpudecomp}
\end{gather}
The result now follows from Proposition~\ref{prop6}.
\end{proof}

We need to point out that unlike in part (e) of Lemma~\ref{properties2} where we set spinors $\psi_k$ to have the same components $\lambda_i$ for the purpose of studying their matrix representation, the actual components of spinors $uf_k, \, k\neq 1,$ are different from those of $uf_1.$ This is because two conjugations are needed in
\begin{gather}
u f_k = m_k (u_k f_1) m_k^{-1}, \qquad u_k = m_k^{-1} u m_k,
\label{eq:ufk}
\end{gather}
for any $m_k \in \cb{M}$ since $f_k = m_k f_1 m_k^{-1}$. This will become evident in our last example.

\begin{example}
Consider again $\cl_{3,0}$ as in Examples~2 and 4. Let $u$ be an arbitrary element in $\cl_{3,0}$ expanded over the  monomial basis 
$$
u = u_1 1 + u_2 \be_{1} + u_3 \be_{2} + u_4 \be_{3} + u_5 \be_{12} + u_6 \be_{13} + u_7 \be_{23} + u_8 \be_{123}. 
$$
Then, matrix $[u]$ in spinor representation in $S_1 = \cl_{3,0}f_1$ is 
\begin{gather}
[u] = \left[\begin{matrix} 
(u_1+u_2)1 + (u_8+u_7)\be_{23} & (u_5+u_3)1-(u_4+u_6)\be_{23}\\[1ex]
(-u_5+u_3)1+(u_4-u_6)\be_{23}  & (u_1-u_2)1+(u_8-u_7)\be_{23}\end{matrix}\right]
\label{eq:[u]}
\end{gather}
whereas matrix $[\tp(u)]$, related to it via Hermitian complex conjugation as predicted by Proposition~\ref{prop7}, is
\begin{gather}
[\tp(u)] = \left[\begin{matrix} 
(u_1+u_2)1 - (u_8+u_7)\be_{23} & (-u_5+u_3)1-(u_4-u_6)\be_{23}\\[1ex]
(u_5+u_3)1+(u_4+u_6)\be_{23}  & (u_1-u_2)1-(u_8-u_7)\be_{23}\end{matrix}\right].
\label{eq:[tpu]}
\end{gather}
Sign reversals in the second column of (\ref{eq:[u]}), i.e., in the components of the second spinor, given that all four non-zero coefficients $\clkj{l}{k}{j}$ are $1$ as shown 
in~(\ref{eq:sumofspinors2}), are caused by the conjugation
$$
u \mapsto \be_2 u \be_2^{-1} = u_1 1 - u_2 \be_{1} + u_3 \be_{2} - u_4 \be_{3} - u_5 \be_{12} + u_6 \be_{13} - u_7 \be_{23} + u_8 \be_{123} 
$$
which reverses signs of four components in $u.$ Observe that commutators of the corresponding basis elements with $\be_2$ are all $-1.$ 
\end{example} 

\section{Action on spinor spaces in semisimple Clifford algebras}\label{action2}
\medskip

Here we summarize only differences in the above results between simple and semisimple algebras. We rely on Theorem~\ref{th:structure} which provides sufficient information about the algebra structure and begin by generalizing Proposition~\ref{prop3}.

As before, let $\cb{F}$ be a complete set of $2N$ mutually annihilating idempotents of the form~(\ref{eq:f}) adding to the unity in a semisimple Clifford algebra $\cl_{p,q},$ $p-q = 1 \bmod 4$. This time, $N=2^{k-1}$ where, as before, $k=q-r_{q-p}$. Since the algebra is non-central, the set $\cb{F}$ partitions into two subsets of $N$ idempotents:
\begin{gather}
\cb{F} = \cb{F}_1 \cup \cb{F}_2 = 
         \{f_1,f_2,\ldots,f_N\} \cup \{\hat{f}_1,\hat{f}_2,\ldots,\hat{f}_N\}
\label{eq:cbF2}
\end{gather}
such that the idempotents in $\cb{F}_1$ (resp., $\cb{F}_2$) add up to say a central idempotent $J_{+}=\frac12(1+\be_{1\cdots n})$ (resp., $J_{-}=\frac12(1-\be_{1\cdots n})$).\footnote{Recall, that 
$\hat{u}$ denotes the grade involution of $u \in \cl_{p,q}$. Then, $\hat{J}_{+}=J_{-}$, $J_{+}J_{-}=J_{-}J_{+}=0$, and $J_{\pm}^2=J_{\pm}$.} Let $f=f_1$ for short and let $G_{p,q}(f)$ be the stabilizer of $f$ in $G_{p,q}$ under the conjugate action. Since the orbit $\cb{O}(f)$ contains now $N=2^{k-1}$ elements, we have
\begin{gather*}
N=[G_{p,q} : G_{p,q}(f)] = |\cb{O}(f)| = |G_{p,q}|/|G_{p,q}(f)| = 2 \cdot 2^{p+q}/|G_{p,q}(f)| = 2^{k-1} 
\end{gather*}
which implies that $G_{p,q}(f) = 2^{2p+r_{q-p}}$ in the semisimple case. Furthermore, it is easy to notice that $\hat{f} =  m \hat{f} m^{-1}$ for every $m \in G_{p,q}(f)$. 

\begin{proposition}
Let $\cl_{p,q}$ be a semisimple Clifford algebra, $p-q = 1 \bmod 4$ and $p+q\leq 9$. Let $f$ be any primitive idempotent in the set $\cb{F}$ and let $G_{p,q}(f)$ be its stabilizer in $G_{p,q}$ under the conjugate action. Then, $G_{p,q}(f)=G_{p,q}(\hat{f})$, and 
\begin{itemize}
\item[(i)] $G_{p,q}(f) \lhd  G_{p,q}$, that is, $G_{p,q}(f)$ is a normal subgroup of $G_{p,q}$ and $|G_{p,q}(f)| = 2^{2+p+r_{q-p}}.$ 
\item[(ii)] $G_{p,q}(f)$ is a $2$-primary Abelian group when $p-q = 0,1,2 \bmod 8$ (real semisimple case) whereas $G_{p,q}(f)$ is a non Abelian $2$-group when $p-q = 4,5,6 \bmod 8$ (quaternionic semisimple case).
\item[(iii)] Let $k = q-r_{q-p}.$ Then, $G_{p,q}(f)$ is generated multiplicatively by $s$ non-unique elements
\begin{gather}
G_{p,q}(f) = \langle g_1,g_2, \ldots, g_s \rangle
\label{eq:gens2}
\end{gather}
where $s=k+1$ when $p-q = 0,1,2 \bmod 8$ whereas $s=k+2$ when $p-q = 4,5,6 \bmod 8$. 
\item[(iv)] The orders of the generators $g_1,g_2,\ldots,g_s$ are $2$ or $4$, and the structure of the stabilizer group $G_{p,q}(f)$ is shown in Tables~4 and 5 in Appendix~\ref{AppendC}.
\item[(v)] Let $m_j$ be an element in $G_{p,q}(f)$ and let $f$ have the form~(\ref{eq:f}).Then, the stability of the idempotent $m_j f m_j^{-1}=f$ implies 
$$
m_j \be_{\iu_1} m_j^{-1} = \be_{\iu_1}, \; m_j \be_{\iu_2} m_j^{-1} = \be_{\iu_2}, \; \ldots, \; m_j \be_{\iu_k} m_j^{-1} = \be_{\iu_k}. 
$$
That is, the set of commuting monomials $\cb{T}= \{\be_{\iu_1},\ldots, \be_{\iu_k}\}$ in $f$ is point-wise stabilized by $G_{p,q}(f).$
\end{itemize}
\label{semisimple} 
\end{proposition}
We should also observe that due to the decomposition 
\begin{gather}
\cl_{p,q} = \cl_{p,q} J_{+} \oplus \cl_{p,q} J_{-}
\label{eq:decomp2}
\end{gather}
into simple ideals with $\cb{F}_1 \subset \cl_{p,q} J_{+}$ and $\cb{F}_2 \subset \cl_{p,q} J_{-}$, the two orbits $\cb{O}(f)$ and $\cb{O}(\hat{f})$ under the conjugate action of $G_{p,q}$ remain disjoint.
\begin{example}
Consider $\cl_{2,1} \cong \Mat(2,\BR) \oplus \Mat(2,\BR)$ with a primitive idempotent
$$
f=\frac14(1+\be_1)(1+\be_{23}).
$$
since $k = 2.$ Then, according to Table~4, we have 
\begin{gather}
G_{2,1}(f) = \langle -1,\be_1,\be_{23} \rangle 
           = \{\pm 1, \pm \be_1, \pm \be_{23}, \pm \be_{123}\} \lhd G_{2,1}
\label{eq:G21}
\end{gather}
and we can choose for a transversal set $\cb{M} = \{1, \be_2\}$. Thus, 
$$
S = \cl_{2,1}f = \spn_\BR \{1f, \be_2 f\}, \qquad 
\hat{S} = \cl_{2,1}\hat{f} = \spn_\BR \{1 \hat{f}, \be_2 \hat{f}\}.
$$
Let $f_1=f$. Then, we have two orbits:
\begin{gather}
\cb{O}(f_1) = \{f_1, f_2 \}, \qquad 
\cb{O}(\hat{f}_1) = \{\hat{f}_1, \hat{f}_2 \},
\label{eq:twoorbits}
\end{gather}
where $f_2 = \be_2 f_1 \be_2^{-1} =  \frac14(1-\be_1)(1-\be_{23})$. Thus, the vee group $G_{2,1}$ partitions into two sets
$$
G_{2,1} = G_{2,1}(f) \cup \be_2 G_{2,1}(f) 
        = \{\pm 1, \pm \be_1, \pm \be_{23}, \pm \be_{123} \} \cup 
          \{\pm \be_2, \pm \be_3, \pm \be_{12}, \pm \be_{13} \}.
$$
Consider now two spinors 
$$
\psi = \psi_1 f_1 + \psi_2 \be_2 f_1 \in S \quad \mbox{and} \quad
\phi = \phi_1 f_1 + \phi_2 \be_2 f_1 \in S
$$
with $\psi_i,\phi_i \in \BR.$ Then, in the spinor representation realized in the double ideal $S \oplus \hat{S}$, we get
$$
[\tp(\psi), \tp(\hat{\psi})] = \left[\begin{matrix} [\psi_1,\psi_1] & [\psi_2,-\psi_2] \\ 
                                                    [0,0] & [0,0] \end{matrix}\right], 
\qquad
[\psi, \hat{\psi}] = \left[\begin{matrix} [\psi_1,\psi_1] & [0,0] \\ [\psi_2,-\psi_2] & [0,0] \end{matrix}\right]
$$
which shows that again the anti involution $\tp$ acts as a transposition. Notice also the inner product on $S$ that is again invariant under the group $G_{2,1}^\ve$:
$$
\tp(\psi)\phi = (\psi_1 \phi_1 + \psi_2 \phi_2) f_1 \in \BK = f_1 \cl_{2,1} f_1 \cong \BR.
$$
\end{example}

In the quaternionic semisimple case when $p-q=4,5,6 \bmod 8$ we can again verify that the anti involution $\tp$ acts as a Hermitian quaternionic conjugation, and that for any spinor 
$\psi \in S = \cl_{p,q}f,$ the Clifford product $\tp(\psi) \psi $ belongs to $f\cl_{p,q}f \cong \BR$.

\section{Conclusions}\label{conclusions}
\medskip

In~\cite{part1} we gave some general arguments why studying the transposition anti-isomorphism is useful and in which way this work can and should be generalized. We do not repeat these arguments here and refer to that paper.

Spinor metrics play an important role in physics since they provide covariant bilinears, which encode physical quantities, see for example chapters 11 and 12 in~\cite{lounesto}. Usually spinor norms employ
Clifford reversion and Clifford conjugation, e.g., the Dirac dagger sending $\psi \mapsto \psi^\dagger = \bar{\psi}^t \gamma_0$. These two conjugations lead to the $\beta_+$ and $\beta_-$ norms if additionally a pure spinor is chosen.

The present paper generalized this setting to bilinear forms $\tp(\psi)\psi$ for any signature. We showed that for Euclidean and anti-Euclidean signatures these forms reduce to $\beta_+$ and~$\beta_-$. In general we get, however, new norms.

We can hence generalize the so-called Salingaros vee groups employing the map $\tp$ to construct the discrete group $G_{p,q}$. The main work done in this paper was to show in detail, how the stabilizer groups $G_{p,q}(f)$ of a primitive idempotent generate, via a transversal for the cosets in $G_{p,q}/G_{p,q}(f)$, a spinor basis, and how the spinor modules for different idempotents are related. The results are summarized in Tables 1 to 5 included in Appendix~A.

Another important result obtained is, that the transposition map $\tp$ of the \emph{real} Clifford algebra induces an (anti) involution of the (double skew) field $\BK$ underlying the spinor modules, \emph{which is not the real field} in general, and hence can be nontrivial. This is a subtle point relating the complex and quaternionic conjugations to the transposition, and by no means in a straight forward way. 

We believe our results, stated only in dimensions $n\leq 9$, are generally true in any  dimension. In order to extend them to all dimensions, one would employ the $\bmod 8$ periodicity of Clifford algebras. However, we would like to do that \emph{including} the full constructive machinery presented here, and that was beyond the aim of the current work.

Having constructed new spinor metrics, a natural question is to identify their invariance groups, as Lounesto did for $\beta_+$ and $\beta_-$. This will be presented in part 3 of this work~\cite{part3}.

\section*{Acknowledgments}
\medskip

Bertfried Fauser would like to thank the Emmy-Noether Zentrum for Algebra at the University of Erlangen for their hospitality during his stay at the Zentrum in  2008/2009.
\vfill
\appendix

\newpage
\section{Data Tables}\label{AppendC}

{\small
\begin{table}[!h]
\label{tb:t1}
\begin{center}
\renewcommand{\arraystretch}{1.5}
\begin{tabular}{ | c | p{1.75in} | c | p{1.5in} |}
\multicolumn{4}{c}
{\bf Table 1: Stabilizer group $G_{p,q}(f)$ of a primitive idempotent $f$ in}\\ 
\multicolumn{4}{c}
{\bf simple Clifford algebra $\cl_{p,q} \cong \Mat(2^k,\BR),\, k=q-r_{q-p}$}\\
\multicolumn{4}{c}
{$p-q \neq 1 \bmod 4, \, p-q = 0,1,2 \bmod 8,\,|G_{p,q}(f)| = 2^{1+p+r_{q-p}}$}\\ 
\hline

\multicolumn{1}{|c|}{$\cl_{p,q}$} & 
\multicolumn{1}{|c|}{$f$}         & 
\multicolumn{1}{|c|}{$G_{p,q}(f)\cong (\BZ_2)^{k+1}$}& 
\multicolumn{1}{|c|}{$|g|$} \\ \hline
$\cl_{1,1}$  
& \multicolumn{1}{|c|}{$\frac12 (1+\be_{12})$} 
& \multicolumn{1}{|c|}{$\langle -1,\be_{12} \rangle \cong (\BZ_2)^2$}
& \multicolumn{1}{|c|}{$(2^2)$} \\ \hline
$\cl_{2,0}$  
& \multicolumn{1}{|c|}{$\frac12 (1+\be_1)$} 
& \multicolumn{1}{|c|}{$\langle -1,\be_{1} \rangle \cong (\BZ_2)^2$}
& \multicolumn{1}{|c|}{$(2^2)$} \\ \hline
$\cl_{2,2}$  
& \multicolumn{1}{|c|}{$\frac14 (1+\be_{13})(1+\be_{24})$} 
& \multicolumn{1}{|c|}{$\langle -1, \be_{13}, \be_{24}\rangle \cong (\BZ_2)^3$}
& \multicolumn{1}{|c|}{$(2^3)$} \\ \hline
$\cl_{3,1}$  
& \multicolumn{1}{|c|}{$\frac14 (1+\be_1) (1+\be_{34})$}
& \multicolumn{1}{|c|}{$\langle -1,\be_{1},\be_{34} \rangle \cong (\BZ_2)^3$}
& \multicolumn{1}{|c|}{$(2^3)$} \\ \hline
$\cl_{0,6}$  
& \multicolumn{1}{|c|}{$\frac18 (1+\be_{123}) (1+\be_{146})(1+\be_{345})$}
& \multicolumn{1}{|c|}{$ \langle -1,\be_{123},\be_{146},\be_{345} \rangle \cong (\BZ_2)^4$}
& \multicolumn{1}{|c|}{$(2^4)$} \\ \hline
$\cl_{3,3}$  
& \multicolumn{1}{|c|}{$\frac18 (1+\be_{14})(1+\be_{25})(1+\be_{36})$}
& \multicolumn{1}{|c|}{$\langle -1,\be_{14}, \be_{25}, \be_{36}\rangle \cong (\BZ_2)^4$}
& \multicolumn{1}{|c|}{$(2^4)$} \\ \hline
$\cl_{4,2}$  
& \multicolumn{1}{|c|}{$\frac18 (1+\be_1)(1+\be_{35})(1+\be_{46})$}
& \multicolumn{1}{|c|}{$\langle -1, \be_{1}, \be_{35},\be_{46} \rangle \cong (\BZ_2)^4$}
& \multicolumn{1}{|c|}{$(2^4)$} \\ \hline
$\cl_{0,8}$  
& $\mbox{\hspace{0.2cm}}\frac{1}{16}  (1+\be_{123}) (1+\be_{146}) \times 
   \newline \mbox{\hspace{1.2cm}} (1+\be_{345}) (1+\be_{367}) $
& \multicolumn{1}{|c|}{$\langle -1,\be_{123},\be_{146},\be_{345},\be_{367}\rangle 
  \cong (\BZ_2)^5$}
& \multicolumn{1}{|c|}{$(2^5)$} \\ \hline
$\cl_{1,7}$  
& $\mbox{\hspace{0.2cm}}\frac{1}{16}  (1+\be_{18}) (1+\be_{234}) \times 
   \newline \mbox{\hspace{1.2cm}} (1+\be_{257}) (1+\be_{456})$
& \multicolumn{1}{|c|}{$\langle -1,\be_{18},\be_{234},\be_{257},\be_{456} \rangle 
  \cong (\BZ_2)^5$}
& \multicolumn{1}{|c|}{$(2^5)$} \\ \hline
$\cl_{4,4}$  
& $\mbox{\hspace{0.2cm}}\frac{1}{16} (1+\be_{15}) (1+\be_{26}) \times 
   \newline \mbox{\hspace{1.2cm}} (1+\be_{37}) (1+\be_{48})$
& \multicolumn{1}{|c|}{$\langle -1,\be_{15},\be_{26},\be_{37},\be_{48} \rangle \cong (\BZ_2)^5$}
& \multicolumn{1}{|c|}{$(2^5)$} \\ \hline
$\cl_{5,3}$  
& $\mbox{\hspace{0.2cm}}\frac{1}{16} (1+\be_1) (1+\be_{36}) \times 
   \newline \mbox{\hspace{1.2cm}} (1+\be_{47}) (1+\be_{58})$
& \multicolumn{1}{|c|}{$\langle -1,\be_{1},\be_{36},\be_{47},\be_{58} \rangle \cong (\BZ_2)^5$}
& \multicolumn{1}{|c|}{$(2^5)$} \\ \hline
$\cl_{8,0}$  
& $\mbox{\hspace{0.2cm}}\frac{1}{16}  (1+\be_1) (1+\be_{2345}) \times 
   \newline \mbox{\hspace{1.2cm}} (1+\be_{2468}) (1+\be_{4567})$
& \multicolumn{1}{|c|}{$\langle -1,\be_{1},\be_{2345},\be_{2468},\be_{4567} \rangle
  \cong (\BZ_2)^5$}
& \multicolumn{1}{|c|}{$(2^5)$} \\ \hline
\multicolumn{4}{l}
{\small Note: The last column lists the orders $(|g_1|, \ldots, |g_s|)$ of the
generators in $\langle g_1,\ldots,g_s \rangle$}\\[-1.5ex]  
\multicolumn{4}{l}
{\small  shown in the third column. Here, $(2^s)$ denotes a sequence 
$\underbrace{(2,\ldots, 2}_{s})$ and $(\BZ_2)^{k+1}$ }\\[-1.5ex]
\multicolumn{4}{l}
{\small  denotes the direct product $\underbrace{\BZ_2 \times \cdots \times \BZ_2}_{k+1}$.}
\end{tabular}
\end{center}
\end{table}
}

\newpage
{\small 
\begin{table}[h]
\label{tab:t2}
\begin{center}
\renewcommand{\arraystretch}{1.5}
\begin{tabular}{ | c | p{1.70in} |c|c|}
\multicolumn{4}{c}
{\bf Table 2: Stabilizer group $G_{p,q}(f)$ of a primitive idempotent $f$ in}\\ 
\multicolumn{4}{c}
{\bf simple Clifford algebra $\cl_{p,q} \cong \Mat(2^k,\BC),\, k=q-r_{q-p}$}\\
\multicolumn{4}{c}
{$p-q \neq 1 \bmod 4, \, p-q = 3,7 \bmod 8, \, |G_{p,q}(f)| = 2^{1+p+r_{q-p}}$}\\ 
\hline

\multicolumn{1}{|c|}{$\cl_{p,q}$} & 
\multicolumn{1}{|c|}{$f$}         & 
\multicolumn{1}{|c|}{$G_{p,q}(f)$}&
\multicolumn{1}{|c|}{$|g|$} \\ \hline
$\cl_{1,2}$%
& \multicolumn{1}{|c|}{$\frac12 (1+\be_{13})$} 
& \multicolumn{1}{|c|}{$\langle \be_{2},\be_{13} \rangle 
  \cong \BZ_2 \times \BZ_4$}
& \multicolumn{1}{|c|}{$(4,2)$} \\ \hline
$\cl_{3,0}$%
& \multicolumn{1}{|c|}{$\frac12 (1+\be_1)$}
& \multicolumn{1}{|c|}{$\langle \be_{1},\be_{23} \rangle
  \cong \BZ_2 \times \BZ_4$}
& \multicolumn{1}{|c|}{$(2,4)$} \\ \hline
$\cl_{0,5}$%
& \multicolumn{1}{|c|}{$\frac14 (1+\be_{123}) (1+\be_{345})$}
& \multicolumn{1}{|c|}{$\langle \be_{3},\be_{12},\be_{45} \rangle 
  \cong (\BZ_4)^3$}
& \multicolumn{1}{|c|}{$(4^3)$} \\ \hline
$\cl_{2,3}$%
& \multicolumn{1}{|c|}{$\frac14 (1+\be_{14}) (1+\be_{25})$}
& \multicolumn{1}{|c|}{$\langle \be_{3},\be_{14},\be_{25} \rangle
  \cong (\BZ_2)^2 \times \BZ_4$}
& \multicolumn{1}{|c|}{$(4,2^2)$} \\ \hline
$\cl_{4,1}$%
& \multicolumn{1}{|c|}{$\frac14 (1+\be_1) (1+\be_{45})$}
& \multicolumn{1}{|c|}{$\langle \be_{1},\be_{23},\be_{45} \rangle
  \cong (\BZ_2)^2 \times \BZ_4$}
& \multicolumn{1}{|c|}{$(2,4,2)$} \\ \hline
$\cl_{1,6}$%
& \multicolumn{1}{|c|}{$\frac18 (1+\be_{17}) (1+\be_{234})(1+\be_{456})$}
& \multicolumn{1}{|c|}{$\langle \be_{4},\be_{17},\be_{23},\be_{56} \rangle \cong \BZ_2 \times (\BZ_4)^{3}$}
& \multicolumn{1}{|c|}{$(4,2,4^2)$} \\ \hline
$\cl_{3,4}$%
& \multicolumn{1}{|c|}{$\frac18 (1+\be_{15}) (1+\be_{26})(1+\be_{37})$}
& \multicolumn{1}{|c|}{$\langle \be_{4},\be_{15},\be_{26},\be_{37} \rangle \cong (\BZ_2)^3 \times \BZ_4$}
& \multicolumn{1}{|c|}{$(4,2^3)$} \\ \hline
$\cl_{5,2}$%
& \multicolumn{1}{|c|}{$\frac18 (1+\be_1) (1+\be_{46}) (1+\be_{57})$}
& \multicolumn{1}{|c|}{$\langle \be_{1},\be_{23},\be_{46},\be_{57} \rangle \cong (\BZ_2)^3 \times \BZ_4$}
& \multicolumn{1}{|c|}{$(2,4,2^2)$} \\ \hline
$\cl_{7,0}$%
& \multicolumn{1}{|c|}{$\frac18 (1+\be_1)(1+\be_{2345})(1+\be_{4567})$}
& \multicolumn{1}{|c|}{$\langle  \be_{1},\be_{23},\be_{45},\be_{67}\rangle \cong \BZ_2 \times (\BZ_4)^{3}$}
& \multicolumn{1}{|c|}{$(2,4^3)$} \\ \hline
$\cl_{0,9}$%
& $\mbox{\hspace{0.2cm}}\frac{1}{16} (1+\be_{123}) (1+\be_{146}) \times%
 \newline \mbox{\hspace{1.1cm}} (1+\be_{345}) (1+\be_{367})$
& $\langle \be_{89},\be_{123},\be_{146},\be_{157},\be_{256}\rangle 
   \cong (\BZ_2)^4 \times \BZ_4$
& \multicolumn{1}{|c|}{$(4,2^4)$} \\ \hline
$\cl_{2,7}$%
& $\mbox{\hspace{0.2cm}}\frac{1}{16}  (1+\be_{18}) (1+\be_{29})  \times 
\newline \mbox{\hspace{1.1cm}} (1+\be_{345}) (1+\be_{567})$
& \multicolumn{1}{|c|}{$\langle \be_{5},\be_{18},\be_{29},\be_{34},\be_{67}\rangle \cong \newline (\BZ_2)^2 \times (\BZ_4)^{3}$}
& \multicolumn{1}{|c|}{$(4,2^2,4^2)$} \\ \hline
$\cl_{4,5}$%
& $\mbox{\hspace{0.2cm}}\frac{1}{16} (1+\be_{16}) (1+\be_{27}) \times 
\newline \mbox{\hspace{1.1cm}} (1+\be_{38}) (1+\be_{49}) $
& \multicolumn{1}{|c|}{$\langle \be_{5},\be_{16},\be_{27},\be_{38},\be_{49} \rangle \cong (\BZ_2)^4 \times \BZ_4$}
& \multicolumn{1}{|c|}{$(4,2^4)$} \\ \hline
$\cl_{6,3}$%
& $\mbox{\hspace{0.2cm}}\frac{1}{16}  (1+\be_1) (1+\be_{47}) \times 
\newline \mbox{\hspace{1.1cm}} (1+\be_{58}) (1+\be_{69}) $
& \multicolumn{1}{|c|}{$\langle \be_{1},\be_{23},\be_{47},\be_{58},\be_{69} \rangle  \cong (\BZ_2)^{4} \times \BZ_4$}
& \multicolumn{1}{|c|}{$(2,4,2^3)$} \\ \hline
$\cl_{8,1}$%
& $\mbox{\hspace{0.2cm}}\frac{1}{16}  (1+\be_1)(1+\be_{89}) \times 
\newline \mbox{\hspace{1.1cm}} (1+\be_{2345})(1+\be_{4567})$
& \multicolumn{1}{|c|}{$\langle \be_{1},\be_{23},\be_{45},\be_{67},\be_{89} \rangle  \cong (\BZ_2)^2 \times (\BZ_4)^{3}$}
& \multicolumn{1}{|c|}{$(2,4^3,2)$} \\ \hline
\end{tabular}
\end{center}
\end{table}
} 
\eject
\newpage

{\small
\begin{table}[h]
\label{tab:t3}
\begin{center}
\renewcommand{\arraystretch}{1.5}
\begin{tabular}{ | c | p{1.75in} | c | p{1.5in} |}
\multicolumn{4}{c}
{\bf Table 3: Stabilizer group $G_{p,q}(f)$ of a primitive idempotent $f$ in}\\ 
\multicolumn{4}{c}
{\bf simple Clifford algebra $\cl_{p,q} \cong \Mat(2^k,\BH),\,k=q-r_{q-p}$}\\
\multicolumn{4}{c}
{$p-q \neq 1 \bmod 4, \, p-q = 4,5,6 \bmod 8, \, |G_{p,q}(f)| = 2^{1+p+r_{q-p}}$}\\ 
\hline

\multicolumn{1}{|c|}{$\cl_{p,q}$} & 
\multicolumn{1}{|c|}{$f$}         & 
\multicolumn{1}{|c|}{$G_{p,q}(f)$}& 
\multicolumn{1}{|c|}{$|g|$} \\ \hline

$\cl_{0,4}$%
& \multicolumn{1}{|c|}{$\frac12 (1+\be_{123})$} 
& \multicolumn{1}{|c|}{$F_3=\langle \be_{1},\be_{2},\be_{3} \rangle$} 
& \multicolumn{1}{|c|}{$(4^3)$} \\ \hline
$\cl_{1,3}$%
& \multicolumn{1}{|c|}{$\frac12 (1+\be_{14})$}
& \multicolumn{1}{|c|}{$\langle \be_{2},\be_{3},\be_{14} \rangle 
\cong F_2 \times \BZ_2$}
& \multicolumn{1}{|c|}{$(4^2,2)$} \\ \hline
$\cl_{4,0}$%
& \multicolumn{1}{|c|}{$\frac12 (1+\be_1)$}
& \multicolumn{1}{|c|}{$\langle \be_{1},\be_{23},\be_{24} \rangle 
\cong F_2 \times \BZ_2$} 
& \multicolumn{1}{|c|}{$(2,4^2)$} \\ \hline
$\cl_{1,5}$%
& \multicolumn{1}{|c|}{$\frac14 (1+\be_{16}) (1+\be_{234})$}
& \multicolumn{1}{|c|}{$\langle \be_{2},\be_{3},\be_{4},\be_{16}
  \rangle \cong F_3 \times \BZ_2$} 
& \multicolumn{1}{|c|}{$(4^3,2)$} \\ \hline
$\cl_{2,4}$%
& \multicolumn{1}{|c|}{$\frac14 (1+\be_{15}) (1+\be_{26})$}
& \multicolumn{1}{|c|}{$\langle \be_{3},\be_{4},\be_{15},\be_{26}
  \rangle \cong F_2 \times (\BZ_2)^{2}$} 
& \multicolumn{1}{|c|}{$(4^2,2^2)$} \\ \hline
$\cl_{5,1}$%
& \multicolumn{1}{|c|}{$\frac14 (1+\be_1)(1+\be_{56})$}
& \multicolumn{1}{|c|}{$\langle \be_{1},\be_{23},\be_{24},\be_{56} \rangle \cong F_2 \times (\BZ_2)^{2}$} 
& \multicolumn{1}{|c|}{$(2,4^2,2)$} \\ \hline
$\cl_{6,0}$%
& \multicolumn{1}{|c|}{$\frac14 (1+\be_1)(1+\be_{2345})$}
& \multicolumn{1}{|c|}{$\langle \be_{1},\be_{23},\be_{24},\be_{25} \rangle  \cong F_3 \times \BZ_2$} 
& \multicolumn{1}{|c|}{$(2,4^3)$} \\ \hline
$\cl_{2,6}$%
& \multicolumn{1}{|c|}{$\frac18 (1+\be_{17})(1+\be_{28})(1+\be_{345})$}
& \multicolumn{1}{|c|}{$\langle \be_{3},\be_{4},\be_{5},\be_{17},\be_{28} \rangle \cong F_3 \times (\BZ_2)^{2}$} 
& \multicolumn{1}{|c|}{$(4^3,2^2)$} \\ \hline
$\cl_{3,5}$%
& \multicolumn{1}{|c|}{$\frac18 (1+\be_{16}) (1+\be_{27}) (1+\be_{38})$}
& \multicolumn{1}{|c|}{$\langle \be_{4},\be_{5},\be_{16},\be_{27},\be_{38} \rangle \cong F_2 \times (\BZ_2)^{3}$}
& \multicolumn{1}{|c|}{$(4^2,2^3)$} \\ \hline
$\cl_{6,2}$%
& \multicolumn{1}{|c|}{$\frac18 (1+\be_1)(1+\be_{57})(1+\be_{68})$}
& \multicolumn{1}{|c|}{$\langle \be_{1},\be_{23},\be_{34},\be_{57},\be_{68} \rangle \cong F_2 \times (\BZ_2)^{3}$} 
& \multicolumn{1}{|c|}{$(2,4^2,2^2)$} \\ \hline
$\cl_{7,1}$%
& \multicolumn{1}{|c|}{$\frac18 (1+\be_1)(1+\be_{78})(1+\be_{2345})$}
& \multicolumn{1}{|c|}{$\langle \be_{1},\be_{23},\be_{24},\be_{25},\be_{78} \rangle \cong F_3 \times (\BZ_2)^{2}$} 
& \multicolumn{1}{|c|}{$(2,4^3,2)$} \\ \hline
\multicolumn{4}{l}
{\small Note: $F_3 =\langle \be_{1},\be_{2},\be_{3}\rangle,$ $|F_3|=16$,
$|\be_{1}|=|\be_{2}|=|\be_{3}|=4$, 
$\be_{1}\be_{2}=-\be_{2}\be_{1}$,
$\be_{1}\be_{3}=-\be_{3}\be_{1}$,}\\[-1.5ex]
\multicolumn{4}{l}
{\small and $\be_{2}\be_{3}=-\be_{3}\be_{2}$. 
$F_2 = \langle \be_{2},\be_{3}\rangle \lhd G_{1,3}(f),$ $|F_2|=8,$
$\be_{2}\be_{3} = -\be_{3}\be_{2},$ $|\be_{2}|=|\be_{3}|=4.$}\\
%
%
\end{tabular}
\end{center}
\end{table}
} 

\newpage
{\small
\begin{table}[!h]
\label{tab:t4}
\begin{center}
\renewcommand{\arraystretch}{1.5}
\begin{tabular}{ | c | p{1.90in} |c|c|}
\multicolumn{4}{c}
{\bf Table 4: Stabilizer group $G_{p,q}(f)$ of a primitive idempotent $f$ in}\\ 
\multicolumn{4}{c}
{\bf semisimple Clifford algebra $\cl_{p,q} \cong \Mat(2^{k-1},\BR) \oplus \Mat(2^{k-1},\BR)$}\\
\multicolumn{4}{c}
{$k=q-r_{q-p},\, p-q = 1 \bmod 4, \, p-q = 0,1,2 \bmod 8, \, |G_{p,q}(f)| = 2^{2+p+r_{q-p}}$}\\ 
\hline

\multicolumn{1}{|c|}{$\cl_{p,q}$} & 
\multicolumn{1}{|c|}{$f$}         & 
\multicolumn{1}{|c|}{$G_{p,q}(f)\cong (\BZ_2)^{k+1}$}& 
\multicolumn{1}{|c|}{$|g|$} \\ \hline
$\cl_{2,1}$  
& \multicolumn{1}{|c|}{$\frac14 (1+\be_{1})(1+\be_{23})$} 
& \multicolumn{1}{|c|}{$\langle -1,\be_{1},\be_{23}\rangle  \cong (\BZ_2)^3$}
& \multicolumn{1}{|c|}{$(2^3)$} \\ \hline
$\cl_{3,2}$  
& \multicolumn{1}{|c|}{$\frac18 (1+\be_1)(1+\be_{24})(1+\be_{35})$} 
& \multicolumn{1}{|c|}{$\langle -1,\be_{1},\be_{24},\be_{35} \rangle \cong (\BZ_2)^4$}
& \multicolumn{1}{|c|}{$(2^4)$} \\ \hline
$\cl_{0,7}$%
& $\mbox{\hspace{0.2cm}}\frac{1}{16}(1+\be_{123})(1+\be_{146})\times \newline 
   \mbox{\hspace{1.2cm}} (1+\be_{345})(1+\be_{367})$
& $\langle -1, \be_{123}, \be_{146}, \be_{345}, \be_{367}\rangle \cong (\BZ_2)^5$
& \multicolumn{1}{|c|}{$(2^5)$} \\ \hline
$\cl_{4,3}$  
& $\mbox{\hspace{0.2cm}}\frac{1}{16} (1+\be_1) (1+\be_{25})\times \newline 
   \mbox{\hspace{1.2cm}} (1+\be_{36}) (1+\be_{47})$
& $\langle -1,\be_1,\be_{25},\be_{36},\be_{47} \rangle \cong (\BZ_2)^5$
& \multicolumn{1}{|c|}{$(2^5)$} \\ \hline
$\cl_{5,4}$  
& $\mbox{\hspace{0.2cm}}\frac{1}{32} (1+\be_1) (1+\be_{26})\times \newline 
   \mbox{\hspace{0.6cm}} (1+\be_{37}) (1+\be_{48})(1+\be_{59})$
& $ \langle -1,\be_{1},\be_{26},\be_{37},\be_{48},\be_{59} \rangle \cong (\BZ_2)^6$
& \multicolumn{1}{|c|}{$(2^6)$} \\ \hline
$\cl_{9,0}$  
& $\mbox{\hspace{0.1cm}}\frac{1}{32} (1+\be_1) (1+\be_{2345})\times \newline 
   \mbox{\hspace{0.6cm}} (1+\be_{2367})(1+\be_{2389})\times \newline 
   \mbox{\hspace{1.2cm}}(1+\be_{2468})$
& $\langle -1,\be_{1}, \be_{2345}, \be_{2367}, \be_{2389},\be_{2468}\rangle \cong (\BZ_2)^6$
& \multicolumn{1}{|c|}{$(2^6)$} \\ \hline
$\cl_{1,8}$  
& $\mbox{\hspace{0.1cm}}\frac{1}{32} (1+\be_1) (1+\be_{2345})\times \newline 
   \mbox{\hspace{0.6cm}} (1+\be_{2367}) (1+\be_{2389}) \times \newline
   \mbox{\hspace{1.2cm}}(1+\be_{2468})$
& $\langle -1,\be_{1}, \be_{2345}, \be_{2367}, \be_{2389},\be_{2468} \rangle \cong (\BZ_2)^6$
& \multicolumn{1}{|c|}{$(2^6)$} \\ \hline
%
\end{tabular}
\end{center}
\end{table}
}

\newpage
{\small
\begin{table}[!h]
\label{tab:t5}
\begin{center}
\renewcommand{\arraystretch}{1.5}
\begin{tabular}{ | c | p{1.90in} |c|c|}
\multicolumn{4}{c}
{\bf Table 5: Stabilizer group $G_{p,q}(f)$ of a primitive idempotent $f$ in}\\ 
\multicolumn{4}{c}
{\bf semisimple Clifford algebra $\cl_{p,q} \cong \Mat(2^{k-1},\BH) \oplus \Mat(2^{k-1},\BH)$}\\
\multicolumn{4}{c}
{$k=q-r_{q-p},\, p-q = 1 \bmod 4, \, p-q = 4,5,6 \bmod 8,\,  |G_{p,q}(f)| = 2^{2+p+r_{q-p}}$}\\ 
\hline

\multicolumn{1}{|c|}{$\cl_{p,q}$} & 
\multicolumn{1}{|c|}{$f$}         & 
\multicolumn{1}{|c|}{$G_{p,q}(f)$}& 
\multicolumn{1}{|c|}{$|g|$} \\ \hline
$\cl_{0,3}$  
& \multicolumn{1}{|c|}{$\frac12 (1+\be_{123})$} 
& \multicolumn{1}{|c|}{$ F_3 = \langle \be_1,\be_2,\be_3\rangle $}
& \multicolumn{1}{|c|}{$(4^3)$} \\ \hline
$\cl_{5,0}$  
& \multicolumn{1}{|c|}{$\frac14 (1+\be_1)(1+\be_{2345})$} 
& \multicolumn{1}{|c|}{$\langle \be_1,\be_{23},\be_{24},\be_{25} \rangle \cong F_3 \times \BZ_2$}
& \multicolumn{1}{|c|}{$(2,4^3)$} \\ \hline
$\cl_{1,4}$  
& \multicolumn{1}{|c|}{$\frac14 (1+\be_{15})(1+\be_{234})$} 
& \multicolumn{1}{|c|}{$\langle \be_2,\be_3,\be_4,\be_{15} \rangle 
\cong F_3 \times \BZ_2$}
& \multicolumn{1}{|c|}{$(4^3,2)$} \\ \hline
$\cl_{2,5}$  
& \multicolumn{1}{|c|}{$\frac18 (1+\be_{16})(1+\be_{27})(1+\be_{345})$} 
& \multicolumn{1}{|c|}{$\langle \be_3,\be_4,\be_5,\be_{16},\be_{27}\rangle \cong F_3 \times (\BZ_2)^2$}
& \multicolumn{1}{|c|}{$(4^3,2^2)$} \\ \hline
$\cl_{6,1}$  
& \multicolumn{1}{|c|}{$\frac18 (1+\be_{1})(1+\be_{67})(1+\be_{2345})$} 
& \multicolumn{1}{|c|}{$\langle \be_1,\be_{23},\be_{24},\be_{25},
\be_{67}\rangle \cong F_3 \times (\BZ_2)^2$}
& \multicolumn{1}{|c|}{$(2,4^3,2)$} \\ \hline
$\cl_{7,2}$  
& \multicolumn{1}{|c|}{$\frac{1}{16} (1+\be_{1})(1+\be_{28})(1+\be_{39})(1+\be_{4567})$} 
& \multicolumn{1}{|c|}{$\langle \be_1,\be_{28},\be_{39},\be_{45},
\be_{56},\be_{57}\rangle \cong F_3 \times (\BZ_2)^3$}
& \multicolumn{1}{|c|}{$(2^3,4^3)$} \\ \hline
$\cl_{3,6}$  
& \multicolumn{1}{|c|}{$\frac{1}{16} (1+\be_{1})(1+\be_{24})(1+\be_{35})(1+\be_{6789})$} 
& \multicolumn{1}{|c|}{$\langle \be_1,\be_{24},\be_{35},\be_{67},
\be_{68},\be_{69}\rangle \cong F_3 \times (\BZ_2)^3$}
& \multicolumn{1}{|c|}{$(2^3,4^3)$} \\ \hline
\multicolumn{4}{l}
{\small Note: $F_3 =\langle \be_{1},\be_{2},\be_{3}\rangle < G_{0,3}(f),$ $|F_3|=16$, $|\be_{1}|=|\be_{2}|=|\be_{3}|=4$, 
              $\be_{1}\be_{2}=-\be_{2}\be_{1}$, $\be_{1}\be_{3}=-\be_{3}\be_{1}$,}\\[-1.5ex]
\multicolumn{4}{l}
{\small and $\be_{2}\be_{3}=-\be_{3}\be_{2}$.}\\
\end{tabular}
\end{center}
\end{table}
} 
\newpage



\begin{thebibliography}{9}
\bibitem{part1} R.~Ab\l amowicz and B.~Fauser, \textit{On the Transposition Anti-Involution in Real Clifford Algebras I: The Transposition Map}, Department of Mathematics, Tennessee Technological University, Technical Report No. 2010-1, May 2010 
%
\bibitem{part3} R.~Ab\l amowicz and B.~Fauser, \textit{On the Transposition Anti-Involution in Real Clifford Algebras III: The Invariance Group of the Transposition Spinor Product} (in preparation) 
%
\bibitem{ablamowicz1998}  R.~Ab\l amowicz, \textit{Spinor Representations of Clifford algebras: A Symbolic Approach}. Computer Physics Communications Thematic Issue - Computer Algebra in Physics Research, \textbf{115}, No. 2--3 (1998) 510--535
\bibitem{ablamfauser2009}
R.~Ab\l amowicz and B.~Fauser, $\mathtt{CLIFFORD}$ for Maple, {\protect \tt http://math.tntech.edu/rafal/} (2009)
\bibitem{hahn} A.~J.~Hahn, \textit{Quadratic Algebras, Clifford Algebras, and Arithmetic Witt Groups}. (Undergraduate Texts in Mathematics) (Springer-Verlag, New York, 1994) 
\bibitem{helmstetterCliffordGroups} J.~Helmstetter, \textit{Groupes de Clifford pour de formes quadratiques de rang quelconque}. C.~R.~Acad.~Sci. Paris, 1977:175--177
\bibitem{helmmicali} J.~Helmstetter and A.~Micali, \textit{Quadratic Mappings and Clifford Algebra}. (Birkh\"{a}user, Basel, 2008)
\bibitem{lounesto} P.~Lounesto, \textit{Clifford Algebras and Spinors}. 2nd ed. (Cambridge University Press, Cambridge, 2001)
\bibitem{passman} D.~S.~Passman, \textit{The Algebraic Structure of Group Rings}. (Robert E. Krieger Publishing Company, 1985)
\bibitem{porteous} I.~R.~Porteous, \textit{Clifford Algebras and the Classical Groups}.  (Cambridge University Press, Cambridge, 1995)
\bibitem{rotman} J.J.~Rotman, \textit{Advanced Modern Algebra}, Revised Printing. (Prentice Hall, Upper Saddle River, 2002)
\bibitem{salingaros1} N.~Salingaros, \textit{Realization, extension, and classification of certain physically important groups and algebras}. J. Math. Phys. \textbf{22} (1981) 226-–232
\bibitem{salingaros2} N.~Salingaros, \textit{On the classification of Clifford algebras and their relation to spinors in $n$ dimensions}. J. Math. Phys. \textbf{23} (1) (1982).
\bibitem{salingaros3} N.~Salingaros, \textit{The relationship between finite groups and Clifford algebras}. J. Math. Phys. \textbf{25} (1984) 738-–742
\bibitem{varlamov} V.~V.~Varlamov, \textit{Universal Coverings of the Orthogonal Groups}. Adv. in Applied Clifford Algebras, \textbf{14}, No. 1 (2004) 81--168
\end{thebibliography}
\end{document}